\documentclass[reprint,superscriptaddress,amsmath,amssymb,aps,pre,nobibnotes]{revtex4-1}

\usepackage{inputenc}
\usepackage{amsmath}
\usepackage{amssymb, amsfonts}
\usepackage{graphicx,epstopdf,bm}
\usepackage[colorlinks,bookmarks,urlcolor=blue,citecolor=blue,linkcolor=blue]{hyperref}
\usepackage[T1]{fontenc}
\usepackage{bbm}

\renewcommand{\vec}[1]{\boldsymbol{#1}}
\newcommand{\pd}[2]{\frac{\partial #1}{\partial #2}}
\newcommand{\PD}[2]{\frac{\partial #1}{\partial #2}}
\newcommand{\D}[2]{\frac{d#1}{d#2}}

\newcommand{\abs}[1]{\left\vert #1 \right\vert}
\newcommand{\inprod}[1]{\left \langle #1 \right \rangle}
\newcommand{\ave}[1]{E[#1]}
\newcommand{\avg}[1]{\avg{#1}}

\newcommand{\lap}{\nabla^2}

\newcommand{\paren}[1]{\left(#1\right)}
\newcommand{\brac}[1]{\left[#1\right]}
\DeclareMathOperator{\prob}{Prob}
\newcommand{\barU}{\bar{U}}

\newcommand{\lt}{p_{\text{\tiny{LT}}}}
\newcommand{\st}{p_{\text{\tiny{ST}}}}
\newcommand{\lst}{\tilde{p}_{\text{\tiny{ST}}}}
\newcommand{\lG}{\tilde{G}}
\newcommand{\lR}{\tilde{R}}
\newcommand{\lambdalt}{\lambda_{\text{\tiny{LT}}}}
\newcommand{\streg}{\phi}
\newcommand{\veta}{\vec{\eta}}
\newcommand{\br}{\vec{r}}
\newcommand{\bro}{\vec{r}_{0}}
\newcommand{\brb}{\vec{r}_{\rm{b}}}
\newcommand{\ax}[2]{#1^{(#2)}}
\newcommand{\axphi}[1]{\ax{\streg}{#1}}
\newcommand{\laxphi}[1]{\ax{\tilde{\streg}}{#1}}

\newcommand{\axrho}[1]{\ax{\st}{#1}}

\newcommand{\axf}[1]{\ax{w}{#1}}
\newcommand{\khat}{\hat{k}} 

\newcommand{\reggf}{\gamma}
\newcommand{\fst}{f_{\text{\tiny{ST}}}}
\newcommand{\flt}{f_{\text{\tiny{LT}}}}

\newcommand{\fltcd}{\mathcal{F}_{\text{\tiny{LT}}}}
\newcommand{\omfree}{\Omega_{\text{free}}}
\DeclareMathOperator{\diam}{diam}
\DeclareMathOperator{\sinc}{sinc}
\def\R{\mathbb{R}} 

\newcommand{\comment}[1]{{{#1}}}

\newcommand{\psibar}{\bar{\Psi}}

\usepackage{amsthm}
\theoremstyle{plain}
\newtheorem{theorem}{Theorem}[section]
\newtheorem{lemma}{Lemma}[section]

\begin{document}

\title{Uniform asymptotic approximation of diffusion to a small target}


\author{Samuel A. Isaacson}
\affiliation{Department of Mathematics and Statistics, Boston University, 111 Cummington Mall, Boston, MA 02215}
\email{isaacson@math.bu.edu}
\author{Jay Newby}
\affiliation{Mathematical Biosciences Institute, Ohio State University, 1735 Neil Avenue, Columbus, OH 43210}
\email{newby@math.utah.edu}

\numberwithin{equation}{section}
\renewcommand{\theequation}{\arabic{section}.\arabic{equation}}


\begin{abstract}
  The problem of the time required for a diffusing molecule, within a
  large bounded domain, to first locate a small target is prevalent in
  biological modeling. Here we study this problem for a small
  spherical target. We develop uniform in time asymptotic expansions
  in the target radius of the solution to the corresponding diffusion
  equation. \comment{Our approach is based on combining expansions of
    a long-time approximation of the solution, involving the first
    eigenvalue and eigenfunction of the Laplacian, with expansions of
    a short-time correction calculated by pseudopotential approximation.}
  These expansions allow the calculation of corresponding expansions
  of the first passage time density for the diffusing molecule to find
  the target.  We demonstrate the accuracy of our method in
  approximating the first passage time density and related statistics
  for the spherically symmetric problem where the domain is a large
  concentric sphere about a small target centered at the origin.
\end{abstract}


\maketitle




\normalsize
\section{Introduction}
Diffusion of a molecule to a spherical trap is a classical problem
important in chemical kinetics.  In an unbounded domain, the problem
reduces to the Smoluchowski theory of reaction kinetics.  In the
context of biological processes, intracellular transport of
biomolecules and chemical reactions occur within closed domains with
complex geometries~\cite{bressloff13a}.  As a first passage time
problem, this is closely related to the narrow escape problem, where a
diffusing molecule escapes a closed domain through a small opening on
the boundary, and the long time behavior has been studied using
matched
asymptotics~\cite{ward93a,condamin07c,schuss07a,coombs09a,pillay10a,cheviakov10b,chevalier11a}.
There are many examples of this type of first passage time problem in
biological modeling, including transport of receptors on the plasma
membrane of a dendrite~\cite{holcman04c,bressloff07a}, intracellular
virus trafficking~\cite{lagache07b}, molecular motor
transport~\cite{bressloff11a}, binding of a transcription factor to a
segment of DNA within a nucleus~\cite{IsaacsonPNAS2011}, and export of
newly transcribed mRNA through nuclear pores~\cite{CullenTRENDS2003}.

Consider a bounded domain $\Omega\subset \R^{3}$, containing a
small, absorbing spherical trap, $\Omega_{\epsilon} \subset \Omega$,
with radius $\epsilon$ centered at $\brb \in \Omega$.  We denote by
$\partial \Omega$ the exterior boundary surface to $\Omega$, and by
$\partial \Omega_{\epsilon}$ the exterior boundary to
$\Omega_{\epsilon}$. The non-trap portion of $\Omega$ is denoted by
$\omfree = \Omega\setminus \{ \Omega_{\epsilon} \cup \partial
\Omega_{\epsilon} \}$.  Consider a molecule undergoing Brownian motion
within $\omfree$. 
We denote by $p(\br,t)$ the probability density that the molecule is
at position $\br \in \omfree$ at time $t$ and has not yet encountered
the trap.  For $D$ the diffusion constant of the molecule, $p(\br,t)$
 satisfies the diffusion equation
\begin{subequations}
\label{eq:50}
  \begin{align}
    \label{eq:1}
    \pd{p}{t} = D\nabla^{2}p(\br,t)&, \quad \br \in \omfree,   t > 0,\\
    \partial_{\veta}p(\br,t) = 0&,  \quad\br \in \partial\Omega,   t > 0,\\
    \label{eq:1c}
    p(\br,t) = 0&, \quad\br \in \partial\Omega_{\epsilon},  t > 0, \\
    \label{eq:6}
   p(\br,0) = \delta(\br-\bro)&, \quad\br \in \omfree,   \bro \in \omfree
  \end{align}
\end{subequations}
where $\partial_{\veta}$ denotes the partial derivative in the outward
normal direction, $\veta$, to the boundary.  Let $T$ label the random
variable for the time at which the molecule first reaches
$\partial\Omega_{\epsilon}$.  \comment{The first passage time
  cumulative distribution is defined as}
\begin{equation}
  \label{eq:63}
  \mathcal{F}(t) \equiv \prob[T<t] = 1 - \int_{\Omega}p(\br, t)  d\br.
\end{equation}
\comment{The solution to~\eqref{eq:50} can be written in terms of an
eigenfunction expansion with}
\begin{equation}
  \label{eq:3}
  p(\br, t) = \sum_{n=0}^{\infty}\psi_{n}(\bro)\psi_{n}(\br)e^{-\lambda_{n} t},
\end{equation}
where the eigenfunctions and eigenvalues satisfy
\begin{subequations}
\label{eq:8}
  \begin{align}
\label{eq:5}
 -D \nabla^{2}\psi_{n}(\br) = \lambda_{n}\psi_{n}&, \quad\br\in\omfree, \\
    \partial_{\veta}\psi_{n} = 0&, \quad\br \in \partial\Omega, \\
    \psi_{n}(\br) = 0&, \quad\br \in \partial\Omega_{\epsilon},
  \end{align}
\end{subequations}
\comment{and the eigenfunctions are orthonormal in
$L^2(\omfree)$.}
We order the eigenvalues so that $0<\lambda_{0}\leq \lambda_{1} \leq
\dots$.  In the limit that the radius of the trap vanishes, the
smallest eigenvalue, subsequently called the principal eigenvalue,
also vanishes (i.e., $\lambda_{0}\to 0$). Similarly, the corresponding
eigenfunction, subsequently called the principal eigenfunction,
approaches $\psi_{0}(\br) \to \frac{1}{\sqrt{\abs{\Omega}}}$ as
$\epsilon \to 0$. Corresponding to these limits, the first passage
time $T\to\infty$ and $\lim_{t\to\infty}\int_{\omfree}p(\br, t)  
d\br = 1$ as $\epsilon\to 0$. In what follows we let $\diam S$ and
$\abs{S}$ denote the diameter and volume of the set $S \subset \R^3$.
For $0 < \epsilon \ll \diam \Omega$ the asymptotics of the principal
eigenvalue are known, and given by $\lambda_{0} \sim \frac{4\pi D
}{\abs{\Omega}}\epsilon$~\cite{ward93a} (see
also~\cite{WardCheviakov2011}).  Note that to first order in
$\epsilon$, $\lambda_{0}$ depends only on the volume of $\Omega$ and
not the domain geometry. Higher order terms which depend on other
properties of the domain are discussed in the next section.

The small $\epsilon$ asymptotics of $\lambda_0$ motivate a large-time
approximation of $p(\br,t)$, based on a separation of time scales.
Truncating the eigenfunction expansion~\eqref{eq:3} after the first
term gives the long time approximation,
\begin{equation}
  \label{eq:47}
  p(\br, t) \sim \frac{1}{\abs{\Omega}}e^{-\frac{4\pi D}{\abs{\Omega}} \epsilon t}, 
  \quad \lambda_{1}t \gg 1,\quad \epsilon \ll \diam \Omega.
\end{equation}
Note, however, that the initial condition~\eqref{eq:6} is not
satisfied by this expansion.  Instead, the initial condition is
modified so that the molecule starts from a uniformly distributed
initial position with
\begin{equation}
  \label{eq:10}
  p(\br, 0) = \frac{1}{\abs{\Omega}}.
\end{equation}
In other words, the long time behavior depends very little on the
initial position of the molecule because it is likely to explore a
large portion of the domain before locating the trap.  
The first passage time density is $ f(t) \equiv \frac{d}{dt}\mathcal{F}(t)$, where $\mathcal{F}(t)$ is given by~\eqref{eq:63}.
The long-time, $\lambda_1 t \gg 1$, approximation of the first passage time density is then
\begin{subequations}
  \begin{align}
  \label{eq:62}
  f(t) \sim \lambda_{0} e^{-\lambda_{0} t}&, \quad    \text{for } \lambda_1 t \gg 1,\\    
  \sim \frac{4 \pi D \epsilon }{\abs{\Omega}} e^{-\frac{4 \pi D}{\abs{\Omega}} \epsilon t}&, \quad \text{for } \epsilon \ll \diam \Omega.
  \end{align}
\end{subequations}
The first passage time is therefore approximately an exponential
random variable, with mean
\begin{subequations}
  \begin{align}
  \label{eq:7}
  \ave{T} \sim \frac{1}{\lambda_{0}}&, \quad  \text{for }\lambda_{1}t \gg 1, \\
  \sim \frac{\abs{\Omega}}{4 \pi D \epsilon}&, \quad\text{for }\epsilon \ll \diam \Omega.
  \end{align}
\end{subequations}
An exponentially-distributed first passage time is an important
assumption in course-grained models, such as the reaction-diffusion
master equation
(RDME)~\cite{GardinerHANDBOOKSTOCH,VanKampenSTOCHPROCESSINPHYS,GardinerRXDIFFME}.
(The RDME is a lattice stochastic reaction diffusion model which
assumes that reacting chemicals are well mixed within a computational
voxel.)  More broadly, exponential waiting times are essential for jump
processes to be Markovian.

The above long-time approximation motivates several questions.  First,
when is the non-exponential, short-time behavior of the first passage
time important?  Second, how does changing the initial position of the
molecule effect the approximation?  It follows from \eqref{eq:7} that
the mean binding time is approximately independent of the initial
position.  On the other hand, as we show here the most likely binding
time, called the mode, depends strongly on the initial position.
\comment{Recently, the importance of the initial position in first
  passage times in confined domains has been studied in the context of
  chemical reactions \cite{Benichou2010gb}, and shown to play a role
  in quantifying the difference between two or more identically
  distributed first passage
  times~\cite{mejia-monasterio11a,mattos12a}}.  More generally, to
estimate spatial statistics for the position of the diffusing molecule
it is necessary to obtain expansions of not just the first passage
time density, $f(t)$, but also the solution to the diffusion equation,
$p(\br,t)$.

\comment{The first passage time problem in a confined domain has also
  been studied from the perspective of a continuous time random walk
  (CTRW) on a finite graph of size $N$~\cite{Benichou2011a}.  Meyer
  and coworkers obtain exact results for the Laplace transform, which
  is the moment generating function for the first passage time
  distribution, and expand the moments for large $N$.  They then
  reconstruct the large $N$ expansion of the first passage time
  distribution from the moments.  These results can also be
  interpreted as an approximation of the first passage time
  distribution in the large volume limit.  This perspective is closely
  related to the one considered here; instead of an expansion in large
  volume, we assume the domain volume is $O(1)$ and expand in terms of
  the small radius of the target.}

Motivated by these and other examples, we develop a uniform in time
asymptotic approximation as $\epsilon \to 0$ of the probability
density, $p(\br,t)$ (see~\eqref{eq:diffDensExpansion}), and the first
passage time density, $f(t)$ (see~\eqref{eq:13}), that accounts for
non-exponential, small time behavior and the initial position of the
molecule.  The paper is organized as follows.  \comment{In
  Section~\ref{sec:unif-asympt-appr} we further develop the long time
  approximation and present the complimentary short time
  \emph{correction} based on a pseudopotential approximation.  Adding
  these two estimates we derive a uniform in time asymptotic expansion
  of $p(\br,t)$ for small $\epsilon$. It must be emphasized that what
  we call the ``short time'' correction is \emph{not} an asymptotic
  expansion of $p(\br,t)$ as $t \to 0$, but instead is a correction
  that when added to the long time expansion for any fixed $t$ and
  $\br$ gives an asymptotic expansion of $p(\br,t)$ in $\epsilon$.}
In Section~\ref{sec:first-passage-time} we use the results of
Section~\ref{sec:unif-asympt-appr} to derive a small $\epsilon$
expansion of the first passage time density (through terms of order
$O(\epsilon^2)$).  Finally, in Section~\ref{sec:spher-trap-conc} these
approximations are compared to the exact solution, exact first passage
time density, and several other statistics for a spherical trap
concentric to a spherical domain.

\section{Uniform asymptotic approximation}
\label{sec:unif-asympt-appr}
\comment{ Our basic approach is to first split $p(\br,t)$ into two
  components: a large time approximation that will accurately describe
  the behavior of $p(\br,t)$ for $\lambda_1 t \gg 1$, and a short time
  correction to this approximation when $\lambda_1 t
  \not\gg 1$. Note, both are defined for all times, but the latter
  approaches zero as $t \to \infty$, and so only provides a
  significant contribution for $\lambda_1 t \not \gg 1$. It should be
  stressed that the short time correction is not an asymptotic
  approximation of $p(\br,t)$ as $t \to 0$, but instead serves as a
  correction to the long time expansion for $\lambda_1 t \not \gg 1$.
}  We write $p(\br,t)$ as
\begin{equation}
  \label{eq:46}
  p(\br,t) = \lt(\br, t) + \st(\br,t),
\end{equation}
\comment{where $\lt$ is the large time approximation and $\st$ is the
  short time correction.  We will take $\lt = \psi_0(\br) \psi_0(\bro)
  \exp\brac{-\lambda_0 t}$ to be the long time approximation of the
  eigenfunction expansion~\eqref{eq:3} of $p(\br,t)$. With this
  choice, $\lt$ and $\st$ satisfy the projected initial conditions}
\begin{align}
  \lt(\br,0) &= \inprod{\psi(\br),\delta(\br-\bro)}\psi(\br) = \psi(\br)\psi(\bro), \label{eq:31}
\\
 \st(\br,0) &=  \delta(\br-\bro)-\psi(\br)\psi(\bro). \label{eq:49}
\end{align}
Here we have dropped the subscript and subsequently identify $\psi$
and $\lambda$ as the principal eigenfunction and eigenvalue
respectively.  Using~\eqref{eq:31} and~\eqref{eq:49} as initial
conditions, and setting $t=0$ in~\eqref{eq:46}, then gives $p(\br,t)
= \delta(\br-\bro)$ as required.

In the next two sections we derive asymptotic expansions of $\lt$ and
$\st$ for $\epsilon \ll \diam \Omega$. The expansion of $\lt$ is based
off the principal eigenvalue and eigenfunction expansions developed
in~\cite{WardCheviakov2011}. The expansion of $\st$ adapts the
pseudopotential method we first used in~\cite{IsaacsonRDMELimsII},
where uniform in time expansions of $p(\br,t)$ and the \comment{first
  passage time cumulative distribution}, $\prob \brac{T < t}$, were
obtained for $\Omega = \R^3$ and $\brb$ the origin. We have found that
a direct pseudopotential approximation of~\eqref{eq:50} in bounded
domains with Neumann boundary conditions breaks down for large, but
finite times.  For example, the direct pseudopotential based expansion
of $\prob \brac{T < t}$ can become negative for large times.
This inaccuracy in the pseudopotential approximation 
  arises from the non-zero steady state solution to the limiting
  $\epsilon=0$ equation.  As we see in
  Section~\ref{sec:unif-asympt-appr-short-time}, by projecting out the
  principal eigenfunction this problem is removed when expanding
  $\st$.  This motivated our use of the splitting $p = \lt + \st$.

\subsection{Large time asymptotic expansion}
\label{sec:unif-asympt-appr-long-time}
Since the initial condition~\eqref{eq:31} is an eigenfunction of the
Laplacian, the long time density is given by
\begin{equation*}
  \lt(\br,t) = \psi(r) \psi(r_0) e^{-\lambda t}.
\end{equation*}
As discussed in the Introduction, there are well known asymptotic
approximations for small $\epsilon$ of $\psi(r)$ and $\lambda$. These
then determine the small $\epsilon$ behavior of $\lt(\br,t)$. 
The expansions of $\psi(r)$ and $\lambda$ are typically given in
terms of the corresponding no-trap problem where $\epsilon = 0$.
Let $G(\br,\br',t)$ denote the fundamental solution to the diffusion
equation in $\Omega$ (\textit{i.e.} the $\epsilon = 0$ problem), then
\begin{align} 
  \label{eq:diffGreensFunct}
  \pd{}{t} G(\br, \br', t) &= D \nabla^{2} G(\br, \br', t), \quad  \br \in \Omega,\\
  G(\br, \br', 0) &= \delta(\br-\br'), \quad  \br \in \Omega, \notag
\end{align}
with the no-flux Neumann boundary condition
\begin{equation}
  \partial_{\veta} G(\br, \br', t)  = 0, \quad  \br \in \partial \Omega.
\end{equation}
We will also need the corresponding solution to the time-independent
problem, the pseudo-Green's function $U(\br,\br')$, satisfying
\begin{equation}
  \label{eq:27}
  D \nabla^{2} U(\br,\br') = \frac{1}{\abs{\Omega}} - \delta(\br-\br'), \quad  \br \in \Omega,
\end{equation}
with the no-flux Neumann boundary condition
\begin{equation}
  \label{eq:28}
  \partial_{\veta} U(\br,\br') = 0, \quad  \br \in \partial \Omega,
\end{equation}
\comment{and the normalization condition
\begin{equation} \label{eq:28b}
  \int_{\Omega} U(\br,\br')   d \br = 0.
\end{equation} 
Within the derivative terms in~\eqref{eq:diffGreensFunct} we can
replace $G(\br,\br',t)$ by $G(\br,\br',t) - \abs{\Omega}^{-1}$.
Integrating the resulting equation in $t$ on $(0,\infty)$, and using
the uniqueness of the solution to~\eqref{eq:27} with the boundary
condition~\eqref{eq:28} and the normalization~\eqref{eq:28b}, we find
}
\begin{equation}
  \label{eq:29}
  \int_{0}^{\infty} \left( G(\br, \br', t) - \frac{1}{\abs{\Omega}}\right)   dt = U(\br, \br').
\end{equation}
Here the term $\abs{\Omega}^{-1}$ is necessary to guarantee convergence
of the integral. Finally, we denote by $\reggf$ the value of the
regular part of $U(\br,\brb)$ at $\br = \brb$,
\begin{equation}
  \label{eq:45}
  \reggf  \equiv \lim_{\br \to \brb}\left[  U(\br,\brb) - \frac{1}{4 \pi D \abs{\br-\brb}}\right].
\end{equation}

Let $\khat = 4 \pi D$. As derived in~\cite{WardCheviakov2011}, the
asymptotic expansions of the principal eigenvalue and eigenfunction
for small $\epsilon$ are
\begin{equation}
  \label{eq:55}
  \lambda \sim \lambdalt \equiv \frac{\khat}{\abs{\Omega}}\left( 1 -  \khat \reggf\epsilon\right)\epsilon,
\end{equation}
and
\begin{align}
  \psi(\br) &\sim \frac{1}{\sqrt{\abs{\Omega}}} + \epsilon   \psi^{(1)}(\br)+\epsilon^{2}  \psi^{(2)}(\br) \notag \\
  & \begin{aligned} \label{eq:48}
    =\frac{1}{\sqrt{\abs{\Omega}}} &\Bigg[ 1 - \epsilon \khat U(\br,\brb )  - \epsilon^2 \khat^2 \Bigg( -\reggf U(\br,\brb ) \\
      &+ \frac{1}{\abs{\Omega}}\int_{\Omega}U(\br,\br')U(\br',\brb)   d\br' \Bigg)\Bigg] + \epsilon^2 \psibar.
  \end{aligned}
\end{align}
Here $\psibar$ denotes the spatial average of the second order term
and is given by~\cite{WardCheviakov2011}
\begin{equation} \label{eq:princEfuncNormalize}
\psibar = -\frac{\khat^2}{2 \abs{\Omega}^{\frac{3}{2}}} \int_{\Omega} \paren{U(\br,\brb)}^2   d \br. 
\end{equation}
Note, the second order term in~\eqref{eq:48} is not explicitly derived
in~\cite{WardCheviakov2011}, but can be found by solving
equation~(2.20) (of~\cite{WardCheviakov2011}) with the
normalization~\eqref{eq:princEfuncNormalize}.  The corresponding
expansion of the initial condition, $\lt(\br,0)$, in $\epsilon$ is
\begin{equation*}
  \lt(\br,0) \sim \frac{1}{\abs{\Omega}} + \axf{1}(\br,\bro)   \epsilon + \axf{2}(\br,\bro)   \epsilon^{2},
\end{equation*}
where the functions $\axf{n}(\br,\bro)$ are obtained from substituting
the asymptotic approximations of the principal eigenfunction and
eigenvalue into~\eqref{eq:31} and collecting terms in $\epsilon$.
We find that 
\begin{equation}
  \label{eq:f1}
  \axf{1}(\br, \bro)   = -\frac{\khat}{\abs{\Omega}} \Big( U(\br,\brb) + U(\bro,\brb) \Big),
\end{equation}
\begin{equation}
  \label{eq:f2}
  \begin{split}
    \axf{2}(\br,\bro) &= \frac{2 \psibar}{\sqrt{\abs{\Omega}}} + \frac{\khat^2}{\abs{\Omega}} U(\br,\brb) U(\bro,\brb) \\
    &\quad+ \frac{\khat^2 \reggf }{\abs{\Omega}} \left[U(\br,\brb) + U(\bro,\brb)\right] \\
    &  - \frac{\khat^2}{\abs{\Omega}^2} \int_{\Omega} \left[U(\br,\br') + U(\bro,\br')\right] U(\br',\brb)   d \br' ,
  \end{split}
\end{equation}
so that the small $\epsilon$ expansion of $\lt(\br,t)$ is then
\begin{equation} \label{eq:ltDensityExpan}
  \lt(\br,t) \sim
  \brac{ \frac{1}{\abs{\Omega}} + \axf{1}(\br,\bro)   \epsilon + \axf{2}(\br,\bro)   \epsilon^{2} }
  e^{-\lambdalt t}.
\end{equation}

\subsection{\comment{Short time correction asymptotic expansion}}
\label{sec:unif-asympt-appr-short-time}
To construct an asymptotic approximation to $\st(\br,t)$ for small
$\epsilon$, \comment{we replace the trap boundary condition 
by a sink term in the PDE involving a Fermi pseudopotential
operator~\cite{FermiPseudoPot,HuangYangPhysRev1957},
subsequently denoted by $V$}. The boundary condition
$\st(\br,t) = 0$ for $\br \in
\partial \Omega_{\epsilon}$ is replaced by the sink term \comment{
\begin{multline}
  \label{eq:57a}
  -V \st(\br,t) \equiv -\epsilon \khat \pd{}{\abs{\br-\brb}} \Big[\abs{\br - \brb} \st(\br,t) \Big]_{\br=\brb}\\
\times\delta(\br-\brb)
\end{multline}
For $r = \abs{\br}$, in the special case that $\brb = \vec{0}$ is the
origin, this reduces to
\begin{equation}
  \label{eq:57}
  -V \st(\br,t) \equiv -\epsilon  \khat \pd{}{r} \brac{r \st(\br,t)}_{\br = \vec{0}} \delta(\br).
\end{equation}

Before we proceed with the pseudopotential approximation, it is
instructive to consider why for an equivalent problem in 1D, replacing
the absorbing boundary condition with a sink term makes the problem
easier to solve.  To see how this idea breaks down in higher
dimensions, and to motivate the pseudopotential operator, we apply the
Laplace transform to \eqref{eq:1} (with $p$ replaced by $\st$ and the
initial condition modified to~\eqref{eq:49}). We replace the Dirichlet
boundary condition~\eqref{eq:1c} by a delta function absorption term
on the right hand side. If $\lst(\br,s)$ denotes the Laplace transform
of $\st(\br,t)$, then we find
\begin{multline}
  \label{eq:17}
  -D \lap \lst + s \lst  = \delta(\br - \bro) - \psi(\br) \psi(\bro)\\
 - C\delta(\br - \brb) \lst,
\end{multline}
so that absorption by the target occurs when the center of the target
is reached (at some rate $C$ that is to be determined).  Using the
Green's function of the diffusion equation~\eqref{eq:diffGreensFunct},
we can write the solution as
\begin{multline}
  \label{eq:18}
  \lst(\br, s) = \tilde{G}(\br, \bro, s) - \psi(\bro)\int_{\Omega}\tilde{G}(\br, \br', s)\psi(\br')d\br'\\
 - C\lst(\brb, s)\tilde{G}(\br, \brb, s).
\end{multline}
To solve the above equation, we need only take the limit $\br \to
\brb$ and solve for $\lst(\brb, s)$.  However, we observe that for
dimensions greater than one $\lim_{\br\to\brb}\tilde{G}(\br, \brb, s)
= \infty$, which is why the naive approach breaks down.  There are a
few different methods for adapting this idea to work in higher
dimensions, namely matched asymptotics \cite{ward93a} and pseudo
potential operators, which is the approach that we use here.}

\comment{The pseudopotential operator, $V$, was developed so that the operator
  \begin{equation*}
    D \lap - V
  \end{equation*}
  on $\Omega$ provides an asymptotic approximation in $\epsilon$ of $D
  \lap$ on $\omfree$ with a zero Dirichlet boundary condition on
  $\partial \Omega_\epsilon$~\cite{HuangYangPhysRev1957}.  It was
  originally constructed for approximating hard core potentials in
  quantum mechanical scattering
  problems~\cite{FermiPseudoPot,HuangYangPhysRev1957}, but has also
  been used in the estimation of diffusion-limited reaction rates in
  two-dimensional periodic systems~\cite{GoldsteinPseudoPotent}. The
  operator was derived in~\cite{HuangYangPhysRev1957} by expanding the
  eigenfunctions~\eqref{eq:8}, $\psi_n(\br)$, in a basis of spherical
  harmonics and then analytically continuing the domain of definition
  of each eigenfunction into the interior of the sphere,
  $\Omega_{\epsilon}$.  On $\Omega$, it was found that formally
\begin{equation*}
  D \lap \psi_n(\br) + \lambda_n \psi_n(\br) = V \psi_n(\br) + O(\epsilon^3).
\end{equation*}
When $\Omega = \R^3$, it has been shown that the asymptotic expansion
for small $\epsilon$ of the solution to the diffusion equation with
pseudopotential interaction agrees with the direct asymptotic
expansion in $\epsilon$ of the exact solution to the diffusion
equation with a zero Dirichlet boundary condition, $p(\br,t) = 0$ for
$\br \in
\partial \Omega_{\epsilon}$, up through terms of order
$O(\epsilon^2)$~\cite{IsaacsonRDMELimsII}.}

The pseudopotential approximation for $\st(\br,t)$ is that
\begin{equation}
  \label{eq:ppModelEq}
  \pd{}{t} \st(\br,t) = D \nabla^{2} \st(\br,t) - V \st(\br,t)
\end{equation}
for $\br \in \Omega$, with the initial condition~\eqref{eq:49} and a
no-flux Neumann boundary condition on $\partial \Omega$. \comment{
  In~\cite{AlbeverioSolvabModelBook,AlbeverioJFunctAnal,AlbeverioSingPerturbBook,DellAntDiffMovPtSrc}
  several approaches are developed for rigorously defining
  pseudopotential interactions (usually called point interactions
  or singular perturbations of the Laplacian in those works).
  Following these works, in particular~\cite{DellAntDiffMovPtSrc,
    IsaacsonRDMELimsII}, we split $\st(\br, t)$ into a regular part,
  $\streg(\br,t)$, and a singular part, $q(t) U(\br,\brb)$, so that
\begin{equation}
  \label{eq:33}
  \st(\br,t) = \streg(\br,t) + q(t) U(\br,\brb).
\end{equation}
Here it is assumed that $\streg(\br,t)$ is ``nice'' as $\br \to \brb$.
In Appendix~\ref{ap:ppSolutRep} we give a more detailed motivation for
this representation.}

To find the asymptotic expansion of $\st(\br,t)$ for small $\epsilon$,
we begin by formulating a closed integral equation for
$\streg(\br,t)$. \comment{As described in~\cite{WardCheviakov2011}, we
  can separate $U(\br,\brb)$ into a component that is regular at $\br
  = \brb$, denoted by $R(\br,\brb)$, and a singular part,
  $\khat^{-1} \abs{\br - \brb}^{-1}$, so that
\begin{equation*}
  U(\br,\brb) = R(\br,\brb) + \frac{1}{\khat \abs{\br - \brb}}.
\end{equation*}
Note that the pseudopotential applied to the singular part of
$U(\br,\brb)$ is zero. The action of the pseudopotential on the
representation~\eqref{eq:33} is therefore
\begin{align*}
  V \brac{ \phi(\br,t) + q(t) U(\br,\brb)} &= \epsilon\khat \brac{ \streg(\brb,t) + \reggf q(t)} \delta(\br-\brb),
\end{align*} 
as $\gamma = R(\brb,\brb)$ by~\eqref{eq:45}.  Substituting the
representation~\eqref{eq:33} of $\st(\br,t)$
into~\eqref{eq:ppModelEq}, we find}
\begin{multline} 
  \label{eq:derivingPhiEq1}
  \pd{\streg}{t} = D \nabla^{2} \streg - \frac{dq}{dt} U(\br,\brb) + q(t)\left( \frac{1}{\abs{\Omega}} 
    - \delta(\br-\brb)\right) \\
  - \epsilon\khat\left( \streg(\brb,t) + \reggf q(t) \right)\delta(\br-\brb).
\end{multline}
We enforce the point boundary condition that the delta function terms
should cancel~\cite{DellAntDiffMovPtSrc,IsaacsonRDMELimsII} so that
\begin{equation} \label{eq:ptBndCondit}
  q(t) = -\frac{\epsilon\khat}{1 + \epsilon\khat \reggf } \streg(\brb,t).
\end{equation}
\comment{After substituting \eqref{eq:ptBndCondit} into \eqref{eq:33} and rearranging terms we find that
\begin{multline}
  \label{eq:19}
  \st(\br, t) = \left(1 - \frac{\epsilon\khat}{1 + \epsilon\khat \reggf }U(\br, \brb) \right)\streg(\brb, t)\\
 + (\streg(\br, t) - \streg(\brb, t)).
\end{multline}
If the starting position is close to the target, the last term on the
right hand side is expected to be small.  It follows that space and
time are approximately decoupled, which is consistent with the results
of the CTRW approach found in \cite{Benichou2011a}.
} 

By~\eqref{eq:ptBndCondit}, equation~\eqref{eq:derivingPhiEq1}
simplifies to
\begin{equation} \label{eq:phiEq}
  \pd{\streg}{t}  = D \nabla^{2} \streg - \frac{dq}{dt} U(\br,\brb) + \frac{1}{\abs{\Omega}} q(t).
\end{equation}
Using Duhamel's principle we find that 
\begin{align}
  \label{eq:phiIntEq1}
  \streg(&\br,t) =  \int_{\Omega}G(\br,\br',t)\streg(\br',0)   d\br'  \notag \\
&- \int_{0}^{t} \int_{\Omega} G(\br,\br',t-s) \left( \frac{dq}{ds}U(\br',\brb) - \frac{q(s)}{\abs{\Omega}} \right)  d\br' ds .
\end{align}
Integrating by parts we find 
\begin{multline}
  \int_{0}^{t} G(\br,\br',t-s) \frac{dq}{ds} ds 
  = \int_{0}^{t} D \nabla^{2} G(\br,\br', t-s) q(s) ds \\
  + q(t)\delta(\br-\br')  - q(0) G(\br,\br',t),
\end{multline}
while the no-flux boundary condition implies
\begin{align}
 & \int_{\Omega} D \nabla^{2} G(\br,\br',t-s) U(\br',\brb)   d \br' \\
  & \quad= \int_{\Omega} G(\br,\br',t-s) D \nabla^{2} U(\br',\brb)    d \br' \\
  &\quad= \int_{\Omega} G(\br,\br',t-s) \left(\frac{1}{\abs{\Omega}} - \delta(\br'-\brb)\right)    d\br' \\
  &\quad= \frac{1}{\abs{\Omega}} - G(\br,\brb,t-s).
\end{align}
Using the two preceding identities, it follows
that~\eqref{eq:phiIntEq1} simplifies to
\begin{align}
  \label{eq:24b}
  \streg(\br,t) &=  -q(t) U(\br,\brb) + \int_{0}^{t}G(\br,\brb, t-s) q(s)   ds \notag \\
  &\phantom{-}\quad + \int_{\Omega}G(\br,\br',t) \brac{ \streg(\br',0) +q(0) U(\br',\brb)}   d\br'.
\end{align}
Eliminating $q(t)$ with the point boundary
condition~\eqref{eq:ptBndCondit} gives
\begin{multline}
  \label{eq:24}
  \streg(\br,t) = \int_{\Omega}G(\br,\br',t)\st(\br',0)   d\br' \\
+  \frac{\khat \epsilon}{1 + \reggf \khat \epsilon} \bigg[U(\br,\brb)\streg(\brb ,t)\\
\left. - \int_{0}^{t}G(\br,\brb, t-s)\streg(\brb ,s)   ds\right].
\end{multline}

We now use the integral equation~\eqref{eq:24} to find an asymptotic
expansion of $\st(\br,t)$ in $\epsilon$. Let
\begin{equation*}
  \st(\br,t) \sim \axrho{0}(\br,t) + \axrho{1}(\br,t)   \epsilon + \axrho{2}(\br,t)   \epsilon^2.
\end{equation*}
Similarly, we define the expansion of the regular part of $\st(\br,t)$ by
\begin{equation*}
  \streg(\br,t) \sim \axphi{0}(\br,t) + \axphi{1}(\br,t)   \epsilon + \axphi{2}(\br,t)   \epsilon^2.
\end{equation*}
Using~\eqref{eq:ptBndCondit} we identify the expansion terms of
$\st(\br,t)$ as
\begin{align*}
  \axrho{0}(\br,t) &= \axphi{0}(\br,t), \\
  \axrho{1}(\br,t) &= \axphi{1}(\br,t) -\khat \axphi{0}(\brb,t) U(\br,\brb), \\
  \axrho{2}(\br,t) &
= \axphi{2}(\br,t)  - \khat \paren{\axphi{1}(\brb,t) - \khat \gamma \axphi{0}(\brb,t)} U(\br,\brb).
\end{align*}
The principal eigenvalue and eigenfunction expansions of the previous
section imply that
\begin{equation*}
  \st(\br,0) \sim \delta(\br - \br_0) - \frac{1}{\abs{\Omega}} - \axf{1}(\br,\bro)   \epsilon - \axf{2}(\br,\bro)   \epsilon^{2}.
\end{equation*}
Substituting this expansion into~\eqref{eq:24} yields
\begin{equation}
  \label{eq:53}
  \axphi{0}(\br,t) = G(\br,\bro,t) - \frac{1}{\abs{\Omega}}, 
\end{equation}
\begin{equation}
  \label{eq:54}
  \begin{split}
    \axphi{1}(\br,t) &=  \khat U(\br,\brb)\axphi{0}(\brb ,t) +\frac{\khat}{\abs{\Omega}} U(\bro,\brb) \\
   &\quad - \khat\int_{0}^{t}G(\br,\brb, t-s)\axphi{0}(\brb ,s)   ds  \\
   &\quad+ \frac{\khat}{\abs{\Omega}} \int_\Omega G(\br,\br',t) U(\br',\brb)  d \br',
   \end{split}
\end{equation}
  and
\begin{widetext}
\begin{multline}
  \label{eq:4}
    \axphi{2}(\br,t) = -\khat^{2} \reggf  U(\br,\brb)\axphi{0}(\brb,t)
    +\khat U(\br,\brb) \axphi{1}(\brb,t)
  + \khat^{2}\reggf  \int_{0}^{t}G(\br,\brb,t-s) \axphi{0}(\brb,s)   ds\\ 
 - \khat \int_{0}^{t}G(\br,\brb,t-s) \axphi{1}(\brb,s)   ds 
    -\int_{\Omega}G(\br,\br',t)\axf{2}(\br',\bro) d\br'.
\end{multline}
\end{widetext}
Evaluating~\eqref{eq:4} requires the calculation of $\axphi{1}(\brb,t) =
\lim_{\br\to\brb}\axphi{1}(\br,t)$. Let $G_0(\br,\br',t) =
G(\br,\br',t)-\frac{1}{\abs{\Omega}}$.  Using~\eqref{eq:29} we have
that
\begin{equation*}
\begin{split}
 & U(\br,\brb) \axphi{0}(\brb,t) - \int_0^t G(\br,\brb,t-s) \axphi{0}(\brb,s)   ds  \\
&\quad =  \int_0^t \brac{\axphi{0}(\brb,t) - \axphi{0}(\brb,t-s) } G_0 (\br,\brb,s)   ds \\
&\quad\qquad+ \axphi{0}(\brb,t) \int_t^{\infty} G_0(\br,\brb,s)   ds\\
&\quad\qquad  - \frac{1}{\abs{\Omega}} \int_0^t \axphi{0}(\brb,s)   ds. 
\end{split}
\end{equation*}
Combining this expression with the identity
\begin{equation*}
  \int_{\Omega}G(\br,\br',t)U(\br',\brb)d\br'  
  =\int_t^{\infty} G_{0}(\br,\brb,s)ds,
\end{equation*}
reusing~\eqref{eq:29}, and taking the limit $\br\to\brb$, we find
\begin{equation}
  \label{eq:97}
  \begin{split}
    &\axphi{1}(\brb, t) = \frac{\khat}{\abs{\Omega}}\int_{t}^{\infty} \axphi{0}(\brb,s)ds\\
    &  + \khat  G(\brb,\bro,t) \int_t^{\infty} G_0(\brb,\brb,s)ds\\
    &- \khat\int_{0}^{t} \bigg[ G(\brb, \bro, t-s) - G(\brb, \bro, t) \bigg] G_0(\brb, \brb, s)ds.
  \end{split}
\end{equation}
Note, in the first integral $G_0(\brb,\brb,s)$ will scale like
$s^{-3/2}$ as $s \to 0$.  This singularity is weakened by the
$G(\brb,\bro,t-s) - G(\brb,\bro,t)$ term, which formally scales like
$s$ as $s \to 0$ (for fixed $t > 0$). As such, the overall singularity
in $s$ is integrable. Similarly, $G(\brb,\bro,t)$ will cancel the
effective singularity in $t$ of the last integral.

We therefore find the recursive expansion formula that
\begin{theorem} \label{thm:shortTimeDensityExpansion}
  The asymptotic expansion of $\st(\br,t)$ for $\epsilon \ll \diam \Omega$
is given by
\begin{subequations} \label{eq:shortTimeDensExpan}
  \begin{align}
    \axrho{0}(\br,t) &= G(\br,\bro,t) - \frac{1}{\abs{\Omega}}, \\
    \begin{split}
      \axrho{1}(\br,t) &= -\khat\int_{0}^{t}G(\br,\brb, t-s)\axphi{0}(\brb ,s)ds \\
      & \quad  + \frac{\khat}{\abs{\Omega}} \int_\Omega G(\br,\br',t) U(\br',\brb)  d \br'\\
      &\quad +\frac{\khat}{\abs{\Omega}} U(\bro,\brb),
    \end{split}\\
    \begin{split}
      \label{eq:st2}
      \axrho{2}(\br,t) &= \khat^{2}\reggf \int_{0}^{t}G(\br,\brb,t-s) \axphi{0}(\brb,s)   ds\\
      & \quad - \khat \int_{0}^{t}G(\br,\brb,t-s) \axphi{1}(\brb,s)   ds \\
      & \quad -\int_{\Omega}G(\br,\br',t)\axf{2}(\br',\bro) d\br'.
    \end{split}
  \end{align}  
\end{subequations}
\end{theorem}

\comment{As a short time correction to $\lt(\br,t)$, we expect as $t
  \to \infty$, $\st(\br,t) \to 0$ (away from the singularity at
  $\br=\brb$). Using that $\lim_{t \to \infty} G(\br,\bro,t) =
  \abs{\Omega}^{-1}$,~\eqref{eq:28b}, and~\eqref{eq:29} it is
  immediate that $\lim_{t \to \infty} \axrho{0}(\br,t) = \lim_{t \to
    \infty} \axrho{1}(\br,t) = 0$ for $\br \neq \brb$. In
  Appendix~\ref{ap:pst2LongTime} we show that $\lim_{t \to \infty}
  \axrho{2}(\br,t) = 0$ for $\br \neq \brb$. }
Let
\begin{align} \label{eq:fstDef}
\fst(t) &\equiv \axphi{0}(\brb,t) = G(\brb, \bro, t)-\frac{1}{\abs{\Omega}},\\
  \barU &\equiv U(\bro, \brb).
\end{align}
Combining Theorem~\ref{thm:shortTimeDensityExpansion} with the long
time expansion~\eqref{eq:ltDensityExpan} we find,
\begin{theorem} \label{thm:diffDensExpansion}
  For $\epsilon \ll \diam \Omega$,
  \begin{equation} 
    \begin{split}
      \label{eq:diffDensExpansion}
      p(\br,t) &\sim  G(\br,\bro,t) - \frac{1}{\abs{\Omega}} \paren{ 1 - e^{-\lambdalt t}}\\ 
      &\quad -\frac{\epsilon \khat}{\abs{\Omega}} \brac{U(\br,\brb)e^{-\lambdalt t} - (1 - e^{-\lambdalt t}) \barU}\\ 
      &\quad + \epsilon \frac{\khat}{\abs{\Omega}} \int_\Omega G(\br,\br',t) U(\br',\brb)   d \br' \\
      &\quad - \epsilon \khat \int_0^t G(\br,\brb,t-s) \fst(s)   ds.
    \end{split}
  \end{equation}
\end{theorem}
Note, based on Theorem~\ref{thm:shortTimeDensityExpansion} one can
derive an expansion of $p(\br,t)$ valid through terms of
$O(\epsilon^2)$. That said, this expression is of sufficient
complexity that we do not summarize it here.

\subsection{First passage time density}
\label{sec:first-passage-time}
Denote by $T$ the first passage time (FPT) for the diffusing molecule
to exit through $\partial \Omega_{\epsilon}$. The \comment{FPT cumulative
distribution is defined as}
\begin{equation}
  \label{eq:32}
  \mathcal{F}(t) \equiv \prob[T<t] =  1-\int_{\Omega} p(\br,t)d\br.
\end{equation}
Substituting \eqref{eq:46} into \eqref{eq:32}, we find that
$\mathcal{F}(t) = 1-\sqrt{\abs{\Omega}}\psi(\bro)e^{-\lambda t} -
\int_{\Omega}\st(\br,t)d\br$, where $\psi$ and $\lambda$ are the
principal eigenfunction and eigenvalue satisfying~\eqref{eq:8} for
$n=0$.  From \eqref{eq:33} it follows that $\int_{\Omega}\st(\br,
t)d\br = \int_{\Omega}\streg(\br, t)d \br$, so that
\begin{align}
  \label{eq:12}
  \mathcal{F}(t) = 1-\sqrt{\abs{\Omega}}\psi(\bro)e^{-\lambda t} - \int_{\Omega}\streg(\br,t)d\br.
\end{align}
\comment{Define the cumulative distribution of a standard exponential
  random variable as
  \begin{equation}
    \label{eq:64}
    Y(\tau) \equiv 1 - e^{-\tau}.
  \end{equation}
  Then, the FPT cumulative distribution corresponding to the leading
  order asymptotic expansion of the long time approximation can be
  written as $\fltcd(t) \equiv Y(\lambdalt t) = 1 - e^{-\lambdalt t}$
  (see Introduction and~\eqref{eq:55}).}  \comment{Since $\lambda =
  O(\epsilon)$, we write the uniform approximation to the FPT
  cumulative distribution in terms of the two time scales $t$ and
  $\tau = \lambda t$. Here $\tau$ denotes a shrunken time-scale.}
Notice from \eqref{eq:53} that at leading order,
$\int_{\Omega}\streg(\br, t)d\br \sim
\int_{\Omega}\axphi{0}(\br,t)d\br = 0$.  Substituting \eqref{eq:48},
\eqref{eq:54}, and \eqref{eq:4} into \eqref{eq:12} and collecting
terms in powers of $\epsilon$ yields $\mathcal{F}(t) \sim
\mathcal{F}_{\epsilon}(t,\lambda t)$, where
\begin{multline} 
  \label{eq:61}
 \mathcal{F}_{\epsilon}(t, \tau) \equiv  \Bigg[1 - \epsilon\khat \barU + 2 \epsilon^2 \psibar \sqrt{\abs{\Omega}} \\ -\epsilon^{2}\khat^{2}\bigg(\frac{1}{\abs{\Omega}}\int_{\Omega} U(\bro, \br')U(\br', \brb)d\br' 
        - \reggf \barU \bigg)\Bigg] Y(\tau) \\
    + \left(\epsilon \khat - \epsilon^{2}\khat^{2} \gamma\right)\int_{0}^{t}\fst(s)   ds 
    + \epsilon^2 \khat \int_0^t \axphi{1}(\brb,s)   ds.
\end{multline}
Here $\barU$ and $\fst(t)$ are defined in~\eqref{eq:fstDef} and
$\axphi{1}(\brb,t)$ is given by~\eqref{eq:97}. In evaluating the
various spatial integrals we have made use of the identities
$\int_{\Omega}G(\br, \br', t)d\br = 1$ and $\int_{\Omega}U(\br, \br') d\br = 0$.
An explicit asymptotic expansion of  $\mathcal{F}(t)$ can then be obtained by using that $\lambda \sim \lambdalt$. 
The uniform approximation of the FPT cumulative distribution is therefore $\mathcal{F}(t) \sim \mathcal{F}_{\epsilon}(t,\lambdalt t)$.

By definition, the FPT density function is $f(t) \equiv
\frac{d}{dt}\mathcal{F}(t)$. We denote the expansion of the long time
scale approximation, $\lambda e^{-\lambda t}$, by
\begin{align}
  \label{eq:15}
  \flt(t) &= \D{}{t}\fltcd(t) = \lambdalt e^{-\lambdalt t} \\
&  = \frac{\epsilon \khat}{\abs{\Omega}}\left( 1 -  \khat \reggf\epsilon\right)   e^{-\frac{\khat}{\abs{\Omega}}\left( 1 -  \khat \reggf\epsilon\right)\epsilon t} 
\end{align}
(see \eqref{eq:55}). \comment{Formally differentiating the asymptotic
  expansion $\mathcal{F}_{\epsilon}(t,\lambdalt t)$, we find}
\begin{theorem}
  The asymptotic expansion of $f(t)$ for $\epsilon \ll \diam \Omega$
  is given by
\begin{multline}   \label{eq:13}
  f(t) \sim \Bigg[1 - \epsilon\khat \barU + 2 \epsilon^2 \psibar \sqrt{\abs{\Omega}}\\
    - \epsilon^{2}\khat^{2}\bigg(\frac{1}{\abs{\Omega}}\int_{\Omega} U(\bro, \br')U(\br', \brb)d\br' - \reggf  \barU\bigg) \Bigg]\flt(t) \\
    + \paren{\epsilon \khat - \epsilon^{2}\khat^{2} \gamma} \fst(t) 
    + \epsilon^2 \khat   \axphi{1}(\brb,t).
\end{multline}
\end{theorem}
Since we have derived the expansion of $f(t)$ by formal
differentiation of the expansion of $\mathcal{F}(t)$, we obtain terms
that are of higher order than $O(\epsilon^2)$ in~\eqref{eq:13} (as
$\flt(t)$ is $O(\epsilon)$).  \comment{However, for brevity we ignore
  the $\epsilon$ dependence of $\lambdalt$ when referring to the order
  of the approximation. In other words, when referring to the
``leading order'', ``first order'', or ``second order'' expansion of
$f(t)$, we mean those terms arising from the derivative of the
corresponding order expansion of $\mathcal{F}_{\epsilon}(t, \lambdalt
t)$, treating $Y(\lambdalt t)$ as $O(1)$.} As such, the ``leading
order'' expansion of $f(t)$ will be $\flt(t)$, the ``first order''
expansion will be
\begin{equation*}
  \paren{1 - \epsilon\khat \barU} \flt(t) 
  + \epsilon \khat   \fst(t),
\end{equation*}
and the ``second order'' expansion will be~\eqref{eq:13}.

\section{A spherical trap concentric to a spherical domain} \label{sec:spher-trap-conc}
To illustrate our asymptotic results we consider the problem of a
diffusing molecule searching for a small spherical trap of radius
$\epsilon$ centered at the origin. We assume the trap is contained
within a larger, concentric, spherical domain with unit radius.  As
this problem is exactly solvable, we will use the exact solution
formulae summarized in this section to study the accuracy of our
asymptotic expansions from the preceding sections as both $\epsilon$
and the number of expansion terms are varied.

Denote by $p(r, t)$ the spherically symmetric probability density for
a diffusing molecule to be a distance $r$ from the origin at time $t$.
We assume the trap is centered at the origin, so that $\brb = \vec{0}$, and
let $r=\abs{\br}$, $r_{0} = \abs{\bro}$.  For $p(r, 0) =
\delta(r-r_{0})/r^{2}$, we have that
\begin{equation}
  \label{eq:35}
  p(r,t) = \iint_{\partial B_{1}(\vec{0})}p(\br, t) dS,
\end{equation}
where $\partial B_{1}(\vec{0})$ denotes the boundary of the unit sphere.

The advantage of this geometry is that an exact solution to the diffusion
equation~\eqref{eq:50} is known~\cite{carslaw59a}. We find
\begin{equation} \label{eq:smolEigFuncExpan}
p(r,t) = \sum_{n=1}^{\infty}{\alpha_n \phi_n(r_0) \phi_{n}(r)e^{-\lambda_n t}}, \quad \epsilon<r<1,
\end{equation}
where 
\begin{equation*}
  \phi_n(r) =\frac{1}{r}\brac{\frac{\sin(\sqrt{\lambda_n}(1-r))}{\sqrt{\lambda_n}}-
  \cos(\sqrt{\lambda_n}(1-r))},
\end{equation*}
$\alpha_n = \int_{\epsilon}^{1} \paren{\phi_n(r)}^2 r^2   dr$, and the eigenvalue
$\lambda_{n}$ is given implicitly by
\begin{equation}
  \label{eq:59}
  \tan^{-1}(\sqrt{\frac{\lambda_{n}}{D}}) - (1-\epsilon)\sqrt{\frac{\lambda_{n}}{D}} + n\pi = 0.
\end{equation}
The corresponding first passage time density is
\begin{equation} \label{eq:52}
  f(t) = - \D{}{t} \int_{\epsilon}^{1} p(r,t) r^2   dr = \frac{2}{r_{0}}\sum_{n=0}^{\infty}b_{n}\lambda_{n}e^{-\lambda_{n}t},
\end{equation}
where
\begin{multline}
  \label{eq:58}
  b_{n} = \Bigg[\sqrt{\frac{D}{\lambda_{n}}}\left(\epsilon-\cos((1-\epsilon)\sqrt{\frac{\lambda_{n}}{D}})\right) \\ +\frac{D}{\lambda_{n}}\sin((1-\epsilon)\sqrt{\frac{\lambda_{n}}{D}})\Bigg]\\
\times\left(\frac{(1 + \frac{\lambda_{n}}{D})\sin((r_{0}-\epsilon)\sqrt{\frac{\lambda_{n}}{D}})}{(1-\epsilon)(1+\frac{\lambda_{n}}{D}) - 1}\right).
\end{multline}

In the remainder of this section, we list the quantities necessary to
compute the asymptotic expansions of $p(r,t)$ and $f(t)$ for small
$\epsilon$. Recalling that $\brb = \vec{0}$, $\bar{U}$
is then given by~\cite{ward93a} 
\begin{equation}
  \label{eq:38}
  \barU = U(\bro, \vec{0}) = \frac{1}{4\pi D}\left(\frac{1}{r_0} + \frac{r_0^{2}}{2} - \frac{9}{5}\right).
\end{equation}
It follows from \eqref{eq:45} that
\begin{equation}
  \label{eq:39}
  \reggf = -\frac{9}{20\pi D},
\end{equation}
and from~\eqref{eq:princEfuncNormalize} that
\begin{equation*}
  \psibar = \frac{-72 \pi}{175 \abs{\Omega}^{\frac{3}{2}}}.
\end{equation*}
The fundamental solution $G(\br_0,\vec{0},t) = g(r_0,0,t) / 4 \pi$,
where $g(r,r_0,t)$ denotes the spherically-symmetric Green's function
for the $\epsilon = 0$ Neumann problem (see
Appendix~\ref{S:neuGFAppendix}), is given by
\begin{align}
  \label{eq:75}
  G(\bro, \vec{0}, t) &= \frac{1}{\abs{\Omega}} + \sum_{n=1}^{\infty}c_{n}e^{-\mu_{n}t}, \\
  G(\vec{0}, \vec{0}, t) &= \frac{1}{\abs{\Omega}} + \sum_{n=1}^{\infty}a_{n}e^{-\mu_{n}t},
\end{align}
where
\begin{equation}
  a_{n} = \frac{1}{2\pi}\left(1+\frac{\mu_{n}}{D}\right),
  \quad  c_{n} = a_{n}\sinc(\sqrt{\frac{\mu_{n}}{D}}r_{0}),
\end{equation}
with $\sinc(x) = \sin(x) / x$.  The eigenvalues, $\mu_{n}$, satisfy
\begin{equation}
  \label{eq:60}
    \tan^{-1} (\sqrt{\frac{\mu_{n}}{D}}) - \sqrt{\frac{\mu_{n}}{D}} + n\pi = 0.
\end{equation}
Note that by comparing \eqref{eq:59} to \eqref{eq:60} it follows that
$\lim_{\epsilon\to0}\lambda_{n}=\mu_{n}$.  Integrating~\eqref{eq:diffDensExpansion} over the unit
sphere we find
\begin{equation}
 \label{eq:sphDiffDensExpansion}
 \begin{split}
   p(r,t) &\sim  g(r,r_0,t) - 3 \paren{ 1 - e^{-\lambdalt t}} \\
   &\quad - 3 \epsilon \khat \brac{U(\br,\vec{0})e^{-\lambdalt t} -  (1 - e^{-\lambdalt t}) \barU}\\ 
   &\quad +  \frac{\epsilon \khat }{\abs{\Omega}} \int_{0}^{1} g(r,r',t) U(\br',\vec{0}) (r')^{2}   dr' \\
    &\quad -   \epsilon \khat\int_0^t g(r,0,t-s) \fst(s)   ds.
 \end{split}
\end{equation}
The asymptotic expansion of the first passage time density, $f(t)$,
can be evaluated directly from~\eqref{eq:13}. Here we use
\eqref{eq:29} and \eqref{eq:75} to express $U(\br, \br')$ as an
eigenfunction expansion. We find that
\begin{equation}
  \label{eq:78}
   \int_{\Omega}U(\bro,\br')U(\br',\vec{0})   d\br'= \sum_{n=1}^{\infty}\frac{c_{n}}{\mu_{n}^{2}}.
\end{equation}
The short time correction to the first passage time density is given
by
\begin{equation}
  \label{eq:34}
  \int_{0}^{t}\fst(s)ds = U(\bro, \vec{0}) - \sum_{n=1}^{\infty}\frac{c_{n}}{\mu_{n}}e^{-\mu_{n}t},
\end{equation}
while
\begin{multline}
  \label{eq:14}
    G(\brb, \bro, t)\int_{t}^{\infty}G_{0}(\brb, \brb, s)ds \\
    = \left(\frac{1}{\abs{\Omega}} + \sum_{n=1}^{\infty}c_{n}e^{-\mu_{n}t}\right)\sum_{m=1}^{\infty}\frac{a_{m}}{\mu_{m}}e^{-\mu_{m}t}.
\end{multline}
To evaluate the time convolution,
\begin{equation*}
\int_{0}^{t} \bigg[ G(\brb, \bro, t-s) - G(\brb, \bro, t) \bigg] G_0(\brb, \brb, s)ds,
\end{equation*}
in~\eqref{eq:97} we use the Python \texttt{quad} routine. The integral
is split into a short time portion, $s \in \paren{0,s^{*}}$, and a
long time portion, $s \in \paren{s^{*},t}$. $s^{*}$ is chosen
sufficiently small that $G(\brb,\brb,s)$ can be approximated by a
Gaussian evaluated at the origin, $(4 \pi D s)^{-3/2}$, with the same
absolute error tolerance we use in evaluating the preceding series
(see Appendix~\ref{S:numericsPoints}).

\subsection{Results} \label{S:numericalResults} We now study the error
between the exact spatial and first passage time densities from the
preceding section, $p(r,t)$ and $f(t)$, and their asymptotic
approximations for small $\epsilon$.  In what follows we keep $R=1$,
$D=1$, and vary $\epsilon$ between $10^{-4}$ and $10^{-1}$. The
tolerances we used in evaluating the various series of the previous
section are given in Appendix~\ref{S:numericsPoints}.

While we are interpreting our spatial and time units as
non-dimensionalized, these choices are also consistent with using
spatial units of $\mu \textrm{m}$ and time units of seconds. With these
units the overall domain has roughly the radius of a yeast cell
nucleus. We may therefore interpret the trap as a DNA binding site
that a diffusing protein is searching for.  While trap radii for DNA
binding sites are not generally experimentally measured, the width of
some DNA binding potentials have been measured. For example, the LexA
protein binding potential was found to have a width of approximately
$.5 \textrm{nm}$~\cite{KuhnerLexADNABond}.

The long time approximation of the first passage time density is the
single exponential $\lambda \exp(-\lambda t)$, with the time-scale
$\lambda^{-1}$. The principal eigenvalue $\lambda$ is given implicitly
by \eqref{eq:59} (with $n=0$) and has the asymptotic approximation
$\lambda \sim \lambdalt$ (see also~\eqref{eq:55}). Hence, for small
$\epsilon$, the long time approximation of the first passage time
density is asymptotic to $\flt(t) = \lambdalt \exp(-\lambdalt t)$. As
described at the end of Section~\ref{sec:first-passage-time}, we refer
to $\flt(t)$ as the leading order approximation of $f(t)$ as $\epsilon
\to 0$. (We will also interchangeably refer to $\flt(t)$ as either the
large time or long time approximation.)

\begin{figure}[tb]
  \centering
  \includegraphics[width=7.5cm]{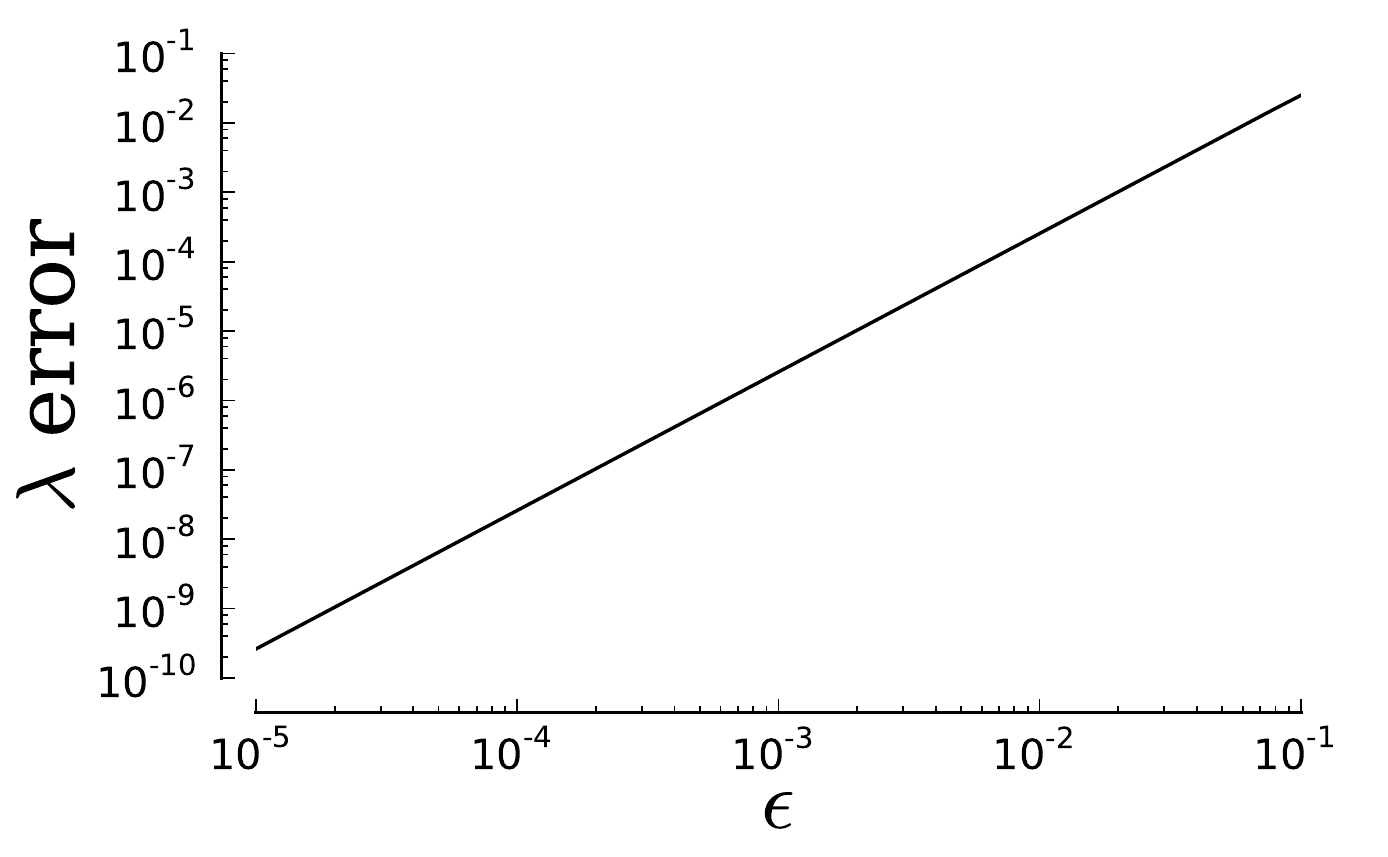}
  \caption{Relative error in approximating the principal eigenvalue,
    $\lambda$, by $\lambdalt$.  Observe that the error decreases like
    $\epsilon^2$ as expected from~\eqref{eq:55}.}
  \label{fig:lambda_error}
\end{figure}
The implicit equation~\eqref{eq:59} can be solved numerically to
calculate $\lambda$ to arbitrary precision by a root finding algorithm
(e.g., Newton's method).  In Fig.~\ref{fig:lambda_error} we compute
the relative error, $\abs{\lambda - \lambdalt} \lambda^{-1}$, of the
asymptotic approximation, $\lambdalt$, as compared to the numerically
estimated value of $\lambda$ (computed to machine precision). We see
that as $\epsilon \to 0$, the relative error between the two decreases
like $\epsilon^{2}$, as expected from~\eqref{eq:55}.

\begin{figure}[tb]
  \centering
  \includegraphics[width=8cm]{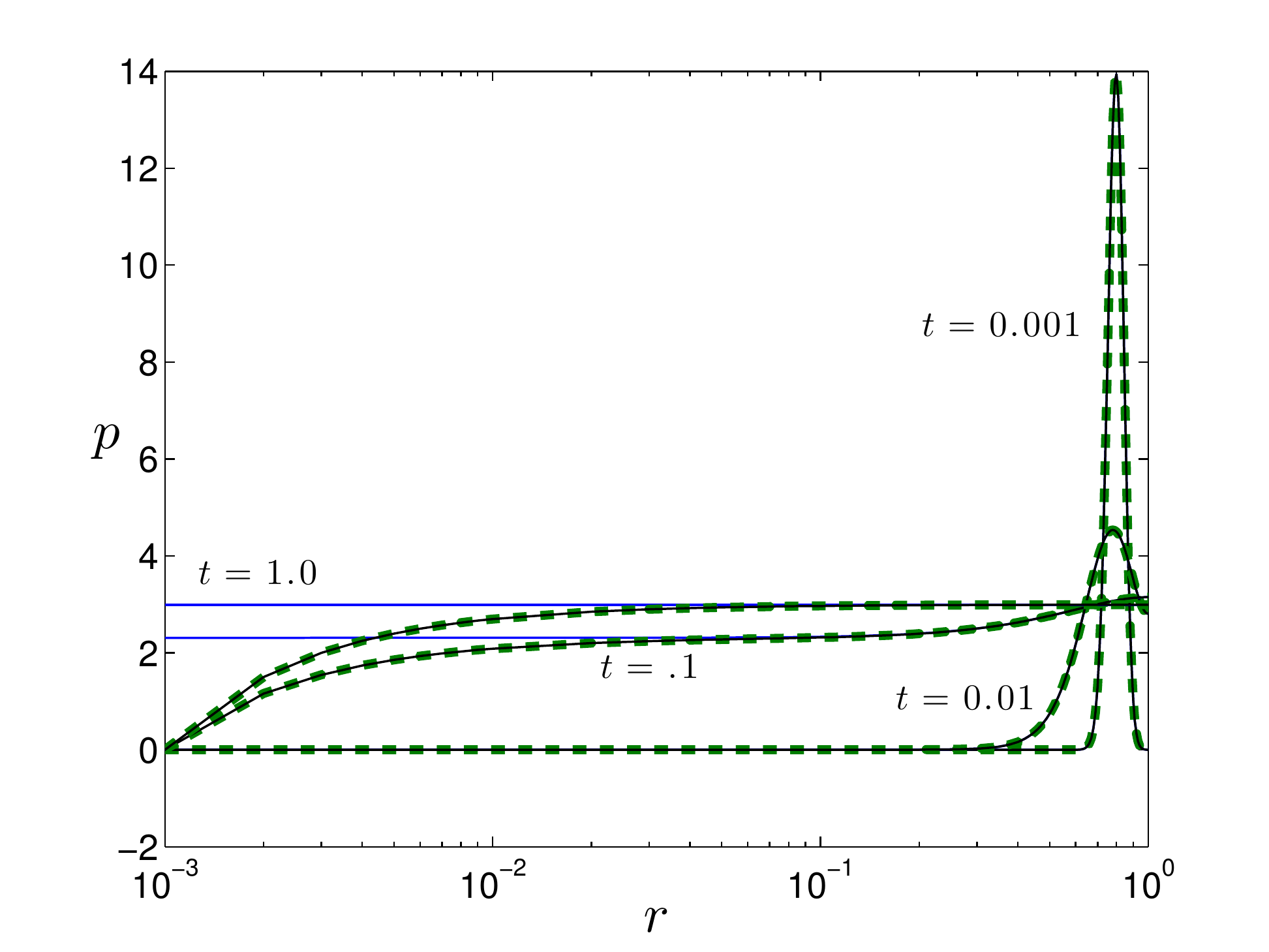}
  \caption{(Color online) The spatial density function $p(r, t)$
    (black curve) and its asymptotic expansions for small $\epsilon$
    at several time points. The blue (light gray) curve gives the
    leading order expansion~\eqref{eq:pLeadingOrderExp},
    $p^{(0)}(r,t)$, while the green dashed curve gives the first order
    expansion~\eqref{eq:sphDiffDensExpansion}. We use a logarithmic
    $x$-axis to emphasize the solution behavior near the target.  $r_0
    = 0.8$ and $\epsilon = 0.001$ (similar to the width of measured
    DNA binding potentials~\cite{KuhnerLexADNABond}). }
  \label{fig:space}
\end{figure}
In Fig.~\ref{fig:space}, we show the leading order spatial density
approximation (blue or light gray curve), the first order expansion
(green dashed curve), and the exact spatial density (black curve).
These curves plot
\begin{equation} \label{eq:pLeadingOrderExp}
  p^{(0)}(r,t) = g(r,r_0,t) - 3 + 3 e^{-\lambdalt t},
\end{equation}
the expansion~\eqref{eq:sphDiffDensExpansion}, and
$p(r,t)$~\eqref{eq:smolEigFuncExpan} respectively.  The spatial
density is shown as a function of $r$ at four different time points.
For this figure we set $r_{0} = 0.8$ and $\epsilon = 0.001$.  The
density is initially concentrated around the initial position at
$t=0.001$ and slowly fills the region $\epsilon<r<1$ until the density
is approximately uniform at $t=1$.  The only visible difference
between the leading order approximation and the exact result is near
the absorbing boundary, $r=\epsilon$, where the exact solution
displays a boundary layer that is lost in the leading order
approximation. The first order
expansion~\eqref{eq:sphDiffDensExpansion} reintroduces this boundary
layer and is indistinguishable from the exact solution at the scale of
the graph.

In the remainder of this section we focus on the approximation of the
first passage time. 
The only free parameters in the model are the radius of the trap,
$\epsilon$, and the initial distance from the trap, $r_{0}$.  The long
time approximation, $\flt(t)$, is independent of $r_{0}$. It follows
that the accuracy of $\flt(t)$ in approximating $f(t)$ improves when
the initial distance from the trap is large (i.e., $\epsilon\ll r_{0}
\leq 1$).  In other words, the long time approximation is best when
the particle is likely to explore a large portion of the domain before
locating the trap.  When the initial distance from the trap is small
(i.e., $\epsilon<r_{0}\ll 1$), we might expect the short time
contribution to be significant since there is a higher probability
that the particle will quickly locate the trap before exploring the
rest of the domain.  In Fig.~\ref{fig:density1}, we show the
asymptotic expansion of the first passage time density~\eqref{eq:13}
for $\epsilon=0.05$ and $r_{0}=0.3$. With this choice the initial
distance of the particle from the trap is small.  Moreover, since the
accuracy of the expansion~\eqref{eq:13} should decrease as $\epsilon$
increases, taking $\epsilon=0.05$ demonstrates the worst case behavior
of the expansion for biologically relevant values of $\epsilon$.

\begin{figure*}[tb]
  \centering
  \includegraphics[width = 15cm]{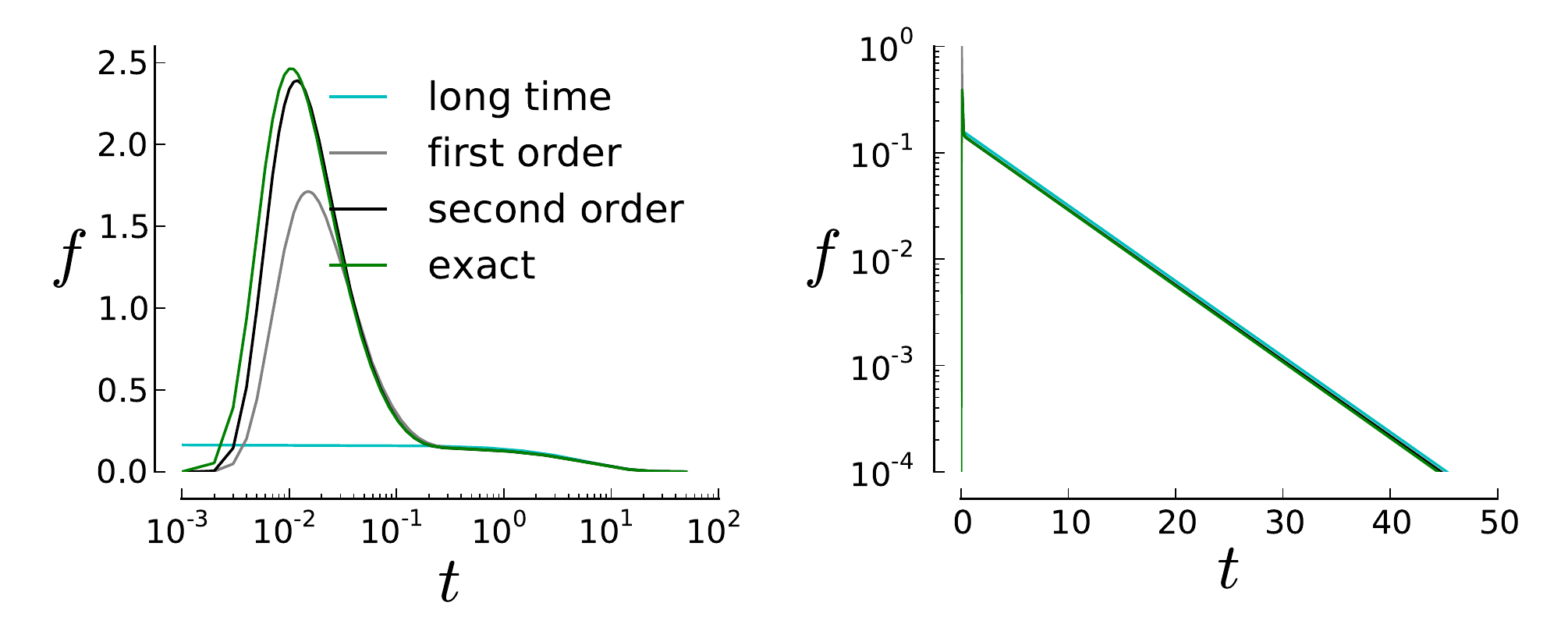}
  \caption{(Color online) The first passage time density, $f(t)$, for
    $r_{0} = 0.3$ and $\epsilon = 0.05$.  Asymptotic approximations of
    varying order are compared to the exact solution.  The left plot
    uses a logarithmic $t$-axis and linear $f$-axis, while the right
    is linear in $t$ and logarithmic in $f$.}
  \label{fig:density1}
\end{figure*}
\begin{figure*}[tbp]
  \centering
  \includegraphics[width = 15cm]{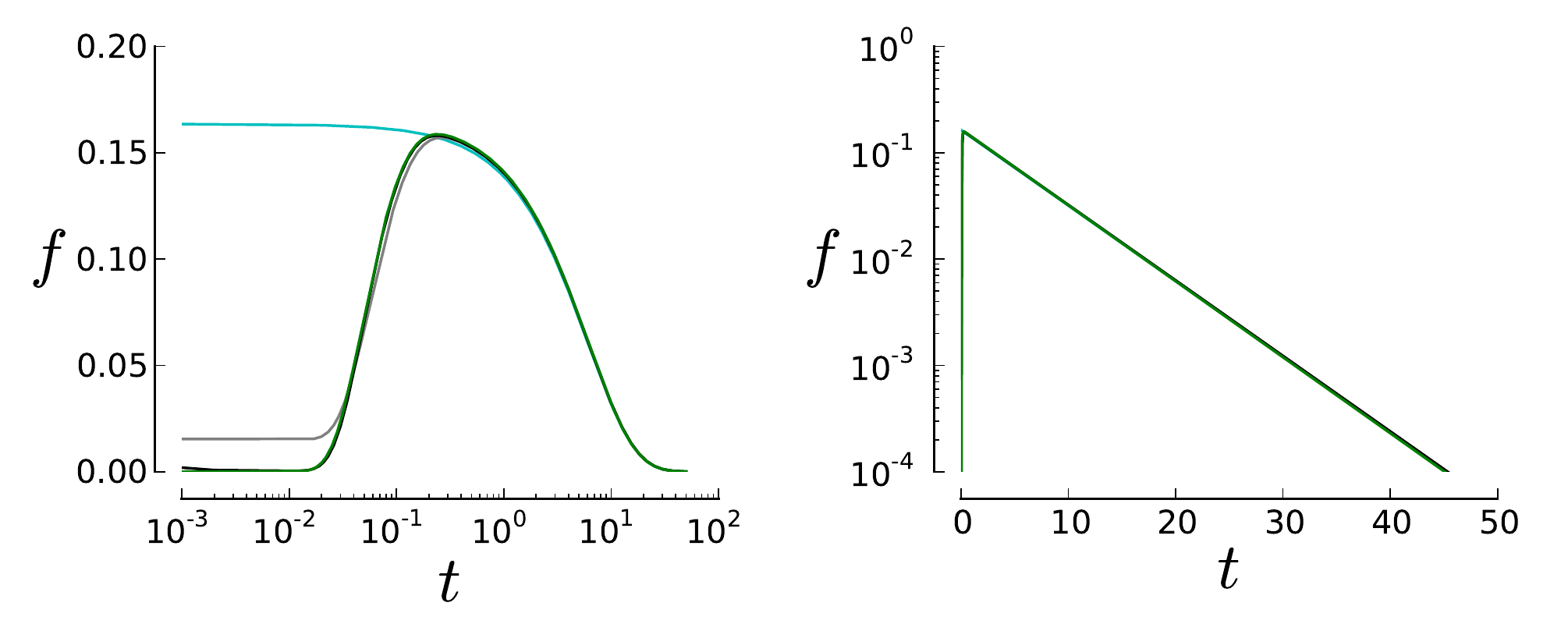}
  \caption{(Color online) The first passage time density, $f(t)$, for
    $r_{0} = 0.8$ and $\epsilon = 0.05$.  Asymptotic approximations of
    varying order are compared to the exact solution. See
    Fig.~\ref{fig:density1} (left panel) for the legend.}
  \label{fig:density2}
\end{figure*}
In Fig.~\ref{fig:density1}(left) the density function is shown with
$t$ on a log scale to accentuate the small time behavior.  There is a
significant difference between the long time approximation (near-flat,
bottom, light blue curve) and the exact solution (uppermost, green
curve). The first and second-order uniform approximations correct for
this difference.  The large-time behavior is shown in
Fig.~\ref{fig:density1}(right) with $f$ on a log scale.  For all
except the shortest times the curve is linear, reflecting the
exponential long-time behavior. We see that on this time-scale there
is very little visible difference between each curve.
Fig.~\ref{fig:density2} is the same as Fig.~\ref{fig:density1}, except
that $r_{0}=0.8$ so that the initial distance from the trap is larger.
In this case the peak in the density occurs at a larger time.  In both
cases, the qualitative difference between the exact solution and the
long-time approximation is a time lag before the exponential long time
behavior dominates.  The time-scale for this time lag is roughly the
diffusive transit time to cover the initial distance from the trap
(i.e., $r_{0}^{2}/D$).

\begin{figure*}[tb]
  \centering
  \includegraphics[width=8cm]{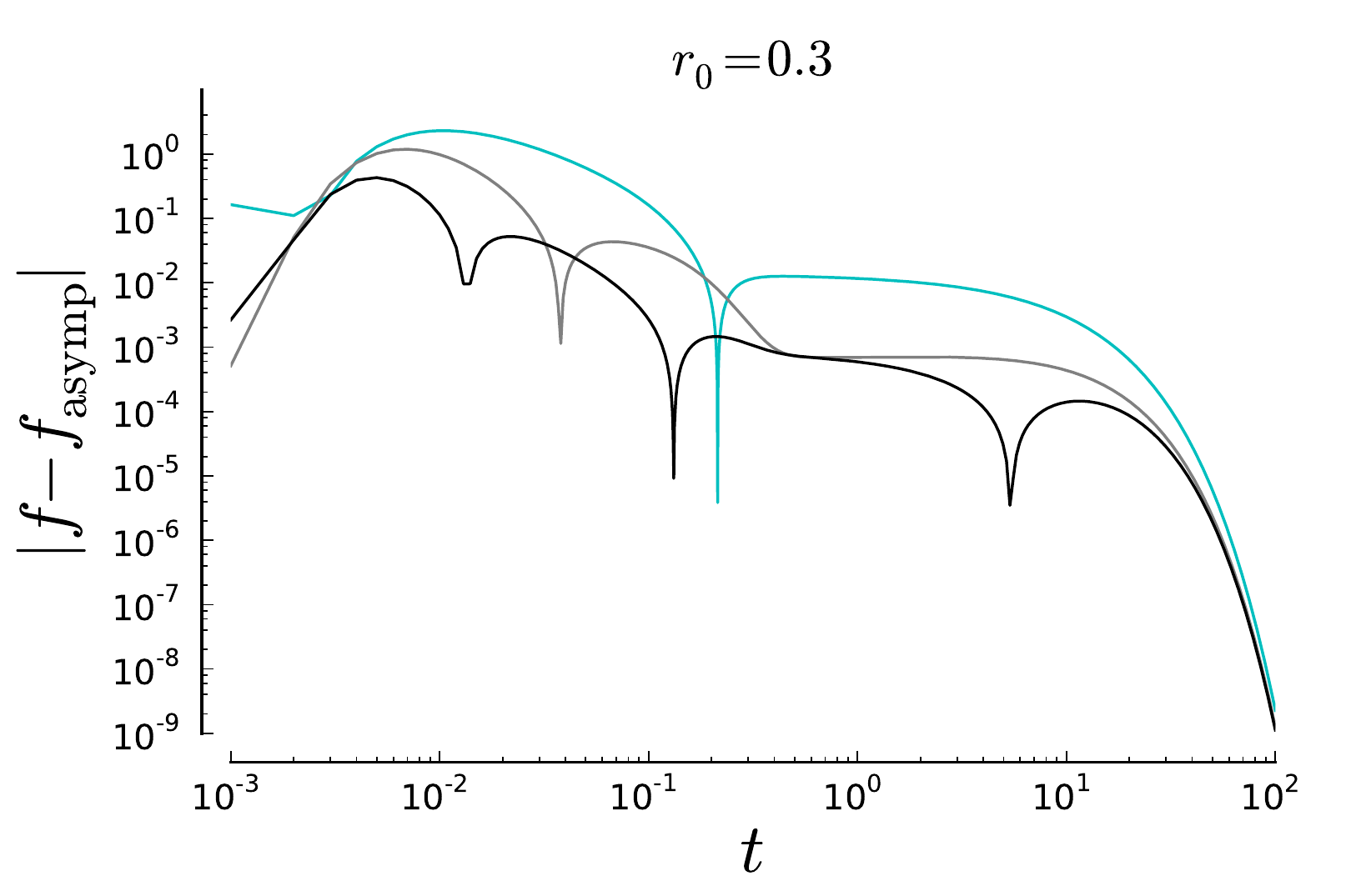}
  \includegraphics[width=8cm]{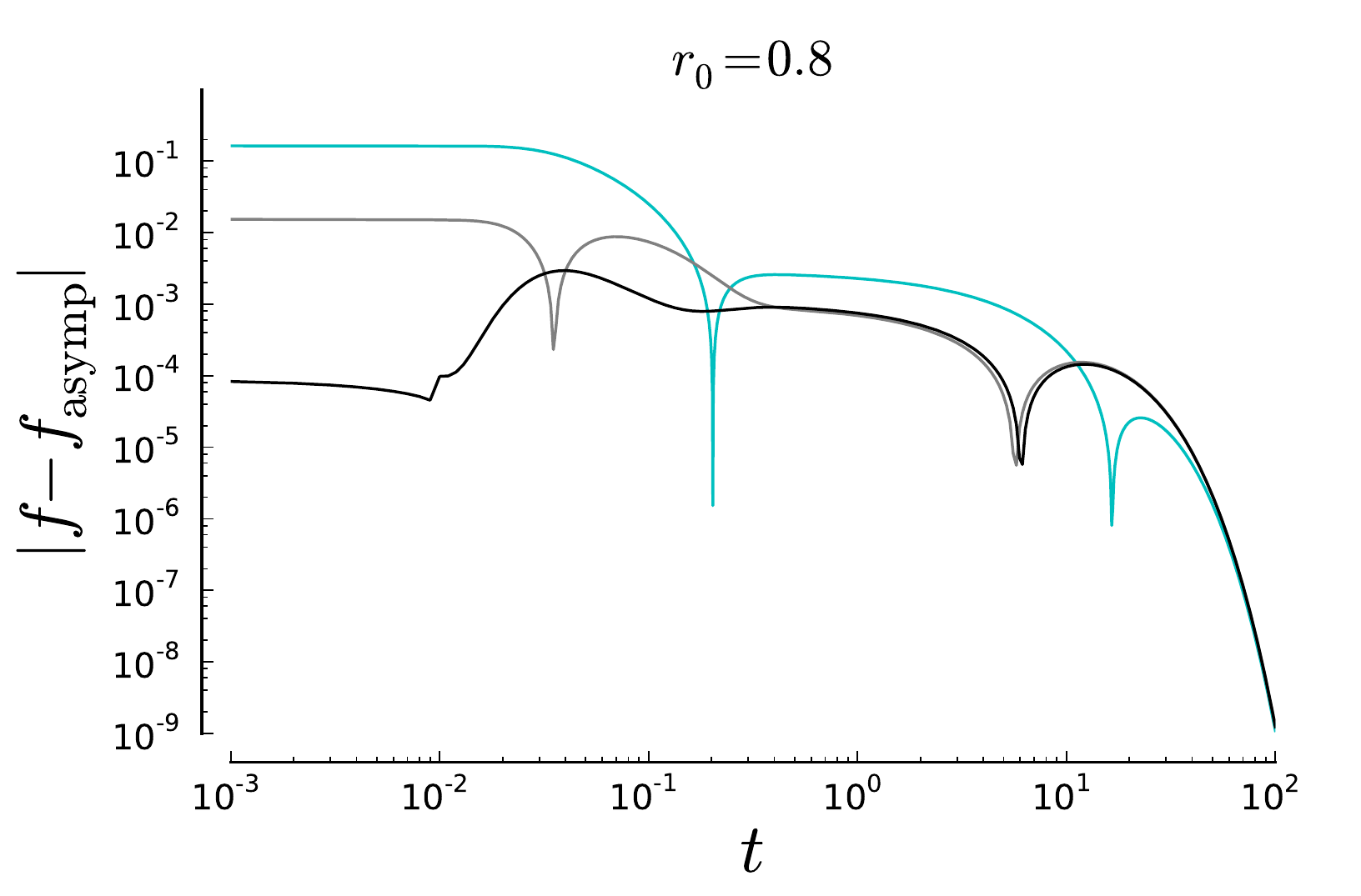}
  \caption{(Color online) Absolute error of the first passage time
    density approximation for $r_{0} = 0.3$ and $r_{0}=0.8$ with
    $\epsilon=0.05$. See Fig.~\ref{fig:density1} (left panel) for the
    legend.}
  \label{fig:density_error1}
\end{figure*}
The absolute error of these approximations is shown in
Fig.~\ref{fig:density_error1} for $r_{0}=0.3$ and $r_{0}=0.8$.  In
both cases, the maximum error is noticeably decreased as the order of
the asymptotic expansion is increased. Comparing the first and second
order expansions, we see the main increase in accuracy results for
times less than $t=1$.  Points in time where one of the approximations
crosses the exact solution result in locally increased accuracy (the
cusp-like drops in the expansion errors).  Interestingly, when $r_{0}
=0.8$ the long time approximation is more accurate for large times
than the first- or second-order uniform approximations.  
Note, however, the error in each expansion at these times is
substantially smaller than for short to moderate times.

\begin{figure}[tb]
  \centering
  \includegraphics[width=8cm]{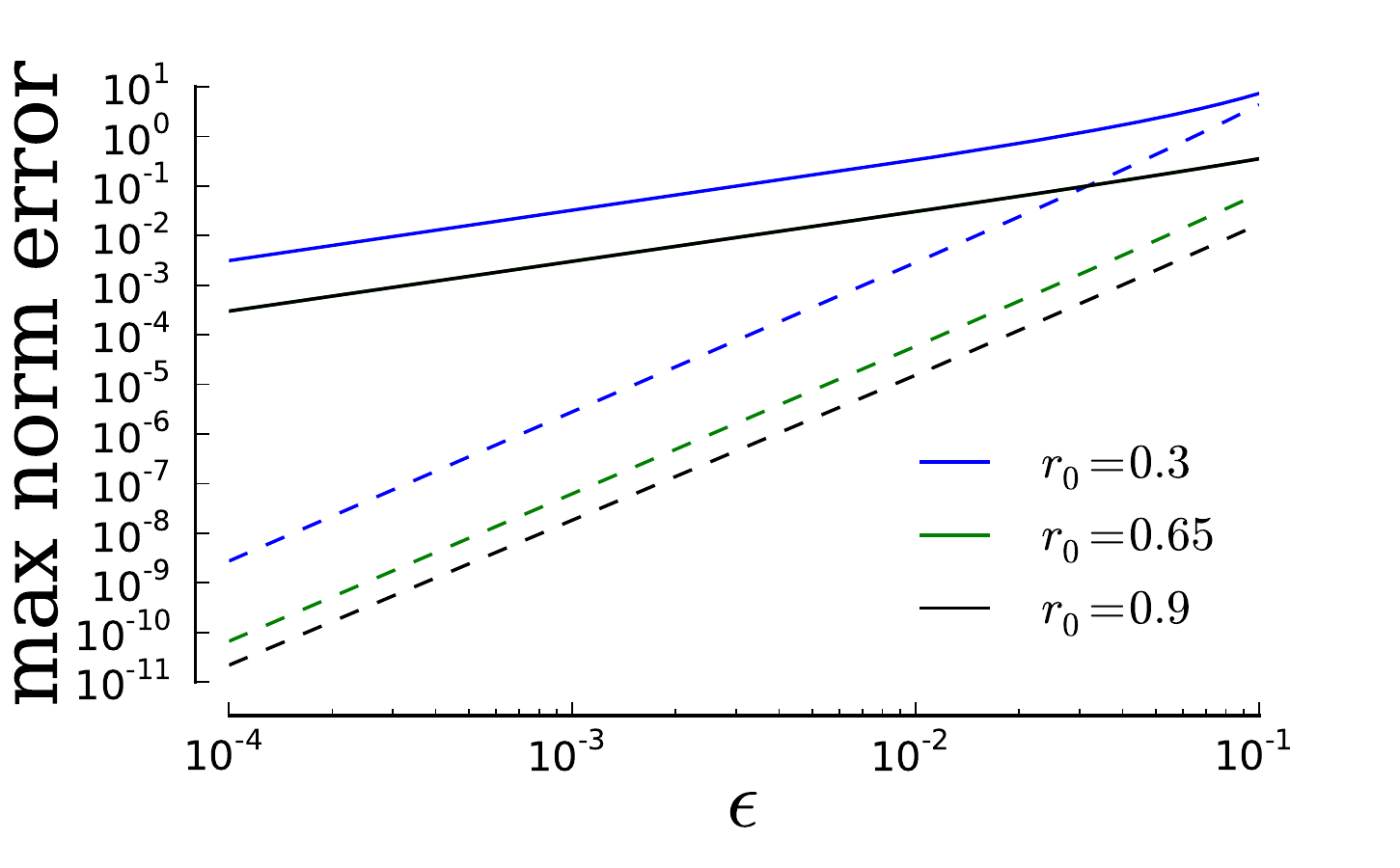}
  \caption{(Color online) The max norm error of the approximation as a
    function of $\epsilon$.  Solid curves show the error of the
    long-time approximation $\flt(t)$. The dashed curves show the
    second order uniform approximation.  Note that the $r_{0}=0.65$
    and $r_{0}=.9$ curves for the large-time approximation are
    indistinguishable.}
  \label{fig:nerror}
\end{figure}
Finally, we examine the max norm error, $\max_{t\geq 0}\abs{f_{\rm
    exact}(t) - f(t)}$, as a function of $\epsilon$ for different
values of $r_{0}$. The time points this error was numerically
evaluated over are the same as those used for the graphs in
Fig.~\ref{fig:density_error1}, and are given in
Appendix~\ref{S:numericsPoints}.  The result shown in
Fig.~\ref{fig:nerror} confirms the asymptotic convergence of the
approximation as $\epsilon \to 0$.  The large-time approximation
\eqref{eq:15} error (solid lines) shows linear convergence, while the
second order uniform approximation \eqref{eq:13} error (dashed line),
which includes short time behavior, shows cubic convergence.

As stated in the Introduction, the mean binding time is well
approximated by the $r_{0}$-independent large-time approximation.
That is, $\ave{T} \sim 1/\lambda$, where $\lambda$ is given by
\eqref{eq:55}.  However, other statistics may be of interest that
depend strongly on $r_{0}$.  One example is the mode, defined as the
most likely binding time, call it $\tau_{\rm m}$, where $f(\tau_{\rm
  m}) = \max_{0\leq t<\infty}f(t)$.  Since the large-time
approximation is an exponential distribution, the corresponding
approximation of the mode is $\tau_{\rm m} \sim 0$.
\begin{figure}[tb]
  \centering
  \includegraphics[width=8cm]{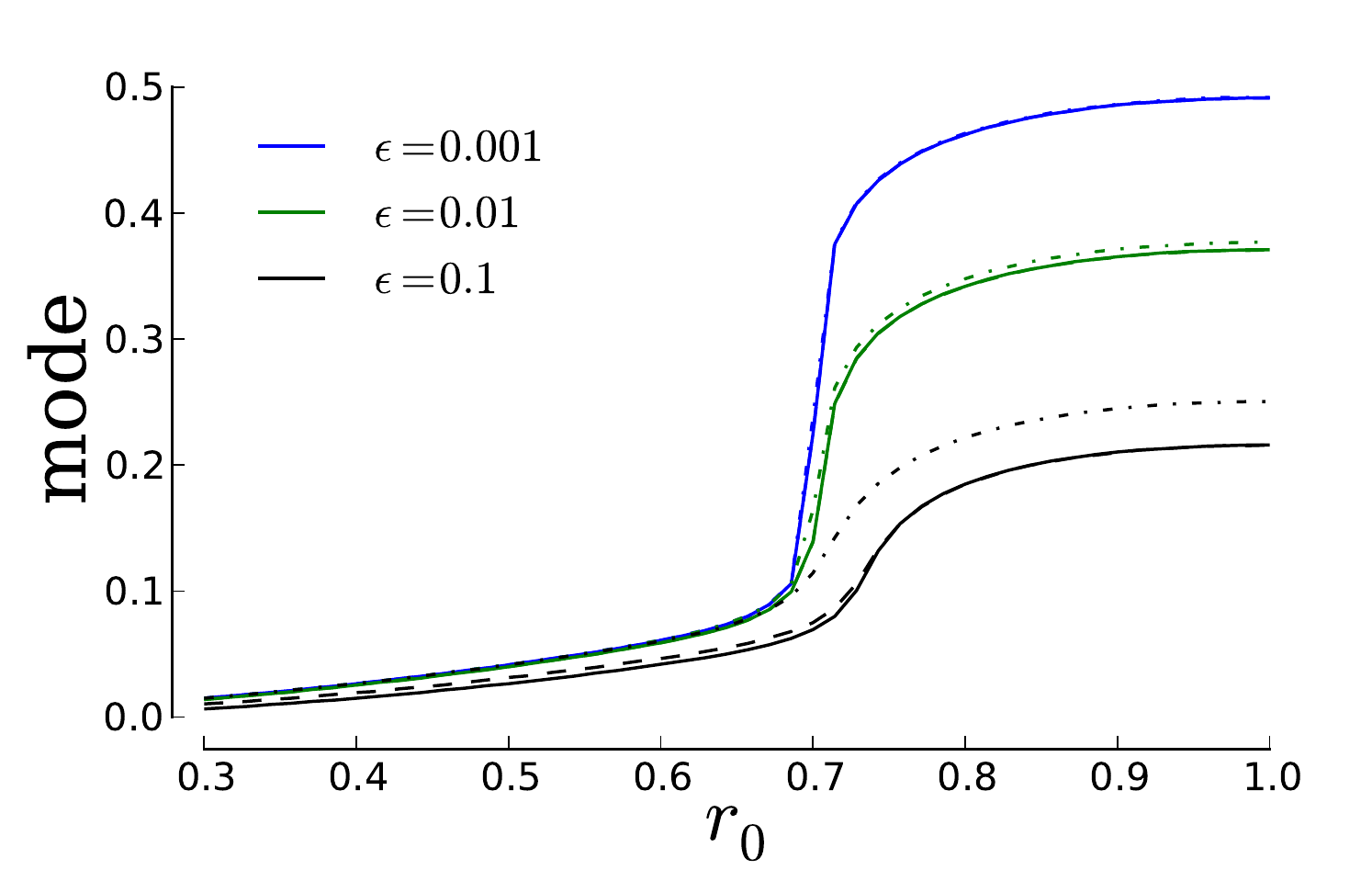}
  \caption{(Color online) The mode of the binding time distribution,
    defined as the most likely binding time, as a function of $r_0$.
    Solid curves show the exact solution, dash dotted curves the first
    order approximation, and the dashed curves the 2nd order
    approximation.}
  \label{fig:mode}
\end{figure}
In figure Fig.~\ref{fig:mode}, we compute the mode by numerically
maximizing the first passage time density.  The exact mode is compared
to first (dash dotted curves) and second order (dashed curves)
approximations of the mode as a function of $r_{0}$.  Each of the
indicated curves are drawn for three different values of $\epsilon$.
For $\epsilon = 10^{-3}$, the difference between each curve is
indistinguishable.  Notice that as $\epsilon$ decreases the mode
increases, particularly for larger values of $r_{0}$, indicating that
the large time approximation of the mode becomes less accurate as
$\epsilon\to 0$.


\section{Discussion}
Although the first passage time of a Brownian particle in a confined
geometry is a well-studied problem, an analytical characterization
that includes short-time behavior of the survival probability density
has been unresolved.  The asymptotic approximation of the long-time
behavior establishes a link between the spatial characteristics of the
problem (i.e., the starting position of the particle and the space
dependent survival probability density) and the short time behavior.
That is, the long time approximation loses information about the
initial position and treats the survival probability density as
uniform in space.  Consequently, the long time approximation is
insufficient if one is interested in statistics that depend on these
spatial characteristics.

\comment{Using a multiple time-scale perturbation approach, we develop
  a long time expansion and a corresponding short time correction to
  this expansion of the solution to the diffusion equation in a
  bounded domain containing a small, absorbing spherical trap. The
  long time approximation is derived from the matched asymptotic
  expansions of~\cite{ward93a}, while the short time correction is
  derived by modification of the pseudopotential method used
  in~\cite{IsaacsonRDMELimsII}. Combining these expansions we develop
  a uniformly accurate (in time) approximation of the survival time
  cumulative distribution and the first passage time density.}  To
study the accuracy of our method, we consider a example problem where
the domain and trap are concentric spheres.  By assuming radial
symmetry, we have available for comparison the exact solution to the
example problem.  Our results show excellent quantitative agreement
for all times over a range of physiologically realistic values of
$\epsilon$.  Moreover, they demonstrate the applicability of our
expansions to estimating statistics that depend critically on the
initial position of the diffusing particle.

  \comment{Our approach should also be applicable to two-dimensional
    systems and multiple targets. Pseudopotentials have already been
    used to approximate rates of diffusion limited reactions in
    two-dimensional periodic systems~\cite{GoldsteinPseudoPotent}.
    Likewise, pseudopotentials were originally developed to study
    many-particle scattering
    problems~\cite{HuangYangPhysRev1957,AlbeverioSolvabModelBook,AlbeverioSingPerturbBook}.
    While we are unaware of their use for approximating first passage
    processes in many-body/target systems, it should be feasible to
    adapt the techniques previously used in the quantum mechanical
    scattering context, allowing the extension of our work to
    multi-target systems.}

\section{Acknowledgments}
SAI was supported by NSF grant DMS-0920886.  JMN was supported in part
by the Mathematical Biosciences Institute and the National Science
Foundation under grant DMS-0931642.  We thank the referees for their
helpful comments and suggestions.

\appendix

\renewcommand{\theequation}{\Alph{section}.\arabic{equation}}
\section{Motivation for assumed form of solution to~\eqref{eq:ppModelEq}} \label{ap:ppSolutRep} 
In~\cite{FaddeevPseudPotDef,AlbeverioJFunctAnal,AlbeverioSolvabModelBook,AlbeverioSingPerturbBook}
several approaches for rigorously defining pseudopotential-like
interactions are presented (usually called point interactions or
singular perturbations of the Laplacian in those works). In the
approach
of~\cite{AlbeverioJFunctAnal,AlbeverioSingPerturbBook,AlbeverioSolvabModelBook},
the Laplacian plus point interaction operator, $D \lap + \alpha
\delta(\br)$, is rigorously constructed so as to be equivalent to the
Laplacian with pseudopotential, $D \lap - V$ (see~\eqref{eq:57a} for
the definition of the pseudopotential, $V$,
and~\cite{AlbeverioSolvabModelBook,AlbeverioJFunctAnal} for details on
the construction of $D \lap + \alpha \delta(\br)$). The
splitting~\eqref{eq:33} is rigorously justified by these works, and in
the context of the diffusion equation goes back at least as far
as~\cite{DellAntDiffMovPtSrc}.

We now give a \emph{formal} motivation for the splitting~\eqref{eq:33}
by studying the Laplace transform of~\eqref{eq:ppModelEq}. Again, we
refer to the
references~\cite{AlbeverioJFunctAnal,AlbeverioSingPerturbBook,AlbeverioSolvabModelBook,DellAntDiffMovPtSrc}
for the rigorous justification.  Our analysis is similar to that given
in Section II.A of~\cite{GrossmanFermippJMP84} (where $\Omega =
\R^3$). Denote by $\tilde{g}(s)$ the Laplace transform of a function,
$g(t)$. Taking the Laplace transform of~\eqref{eq:ppModelEq} we find
\begin{align*}
  -D \lap \lst(\br,s) + s \lst(\br,s) = &-V \lst(\br,s) + \delta(\br - \bro) \\&- \psi(\br) \psi(\bro), 
 \end{align*}
for $\br \in \Omega$ and  $s > 0$,
with the Neumann boundary condition that $\partial_{\veta} \lst(\br,s)
= 0$ for $\br \in \partial \Omega$. We assume the coefficient of $\delta(\br-\brb)$
within the pseudopotential term is finite, and subsequently denote it by 
$B(s)$ (as in~\cite{GrossmanFermippJMP84}),
\begin{align*}
  -B(s) \delta(\br - \brb) &\equiv -V \lst(\br,s) \\
  &= -\epsilon   \khat   \pd{}{\abs{\br-\brb}} \Big[ \abs{\br - \brb} \lst(\br,s) \Big]_{\br = \brb}\\
  &\phantom{=} \quad \times \delta(\br-\brb).
\end{align*}
Recalling that $G(\br,\bro,t)$ is the Green's function for the
$\epsilon=0$ problem~\eqref{eq:diffGreensFunct}, we may then write
\begin{align*}
  \lst(\br,s) &= \lG(\br,\bro,s) - \psi(\bro) \int_{\Omega} \lG(\br,\br',s) \psi(\br')    d \br'\\
  &\phantom{=} \quad- B(s) \lG(\br,\brb,s) \\
  &= H(\br,s) - B(s) \lG(\br,\brb,s),
\end{align*}
where $H(\br,s)$ subsequently denotes the first two terms.
Substituting the preceding equation for $\lst(\br,s)$ into the
definition of $B(s)$ we find
\begin{equation*}
  B(s) = \epsilon \khat \brac{ H(\brb,s) -B(s) \paren{ \lR(\brb,\brb,s) - \frac{1}{\khat} \sqrt{\frac{s}{D}}}}.
\end{equation*}
Here we have split $\lG(\br,\br',s)$ into a part that is regular at
$\br=\br'$, $\lR(\br,\br',s)$, and an explicit singular part so that
\begin{equation*}
  \lG(\br,\br',s) = \lR(\br,\br',s) + \frac{e^{-\abs{\br-\br'} \sqrt{\frac{s}{D}}}}{\khat \abs{\br-\br'}}.
\end{equation*}
Solving for $B(s)$ we find
\begin{equation*}
  B(s) = \frac{\epsilon \khat H(\brb,s)}{1 + \epsilon \khat \lR(\brb,\brb,s) - \epsilon \sqrt{\frac{s}{D}}}.
\end{equation*}
Here we see how the pseudopotential corrects the naive point sink
approximation, as given by~\eqref{eq:18}. The addition of the radial
derivative in the definition of $V$ allows the pseudopotential to
remove $r^{-1}$ type singularities in three-dimensions. This allows
the unknown coefficient, $B(s)$, to be determined.

Using the last equation for $B(s)$, we find that
\begin{equation*}
  \lst(\br,s) = H(\br,s) - \frac{\epsilon \khat H(\brb,s) \lG(\br,\brb,s)}{1 + \epsilon \khat \lR(\brb,\brb,s) - \epsilon \sqrt{\frac{s}{D}}}.
\end{equation*}
We may write
\begin{equation} \label{eq:ppSolutRepLT}
  \lst(\br,s) = \tilde{\phi}(\br,s) + \tilde{q}(s) U(\br,\brb),
\end{equation}
where 
\begin{equation*}
  \tilde{q}(s) = - \frac{\epsilon \khat H(\brb,s)}{1 + \epsilon \khat \lR(\brb,\brb,s) - \epsilon \sqrt{\frac{s}{D}}}
\end{equation*}
and
\begin{equation*}
  \tilde{\phi}(\br,s) = H(\br,s) + \tilde{q}(s) \paren{\lG(\br,\brb,s) - U(\br,\brb)}.
\end{equation*}
As the singular part of $U(\br,\brb)$ is $\khat^{-1}
\abs{\br-\brb}^{-1}$~\cite{WardCheviakov2011}, $\tilde{\phi}(\br,s)$
is regular at $\br = \brb$ for $s > 0$. \emph{Formally}, taking an
inverse Laplace transform of~\eqref{eq:ppSolutRepLT} gives the
representation~\eqref{eq:33} of $\st(\br,t)$.

\section{Limit as $t \to \infty$ of $\axrho{2}$}
\label{ap:pst2LongTime} In this appendix we show that as $t \to
\infty$, $\axrho{2}(\br,t) \to 0$ for $\br \neq \brb$. As in the last
appendix, $\tilde{g}(s)$ will denote the Laplace transform of a
function, $g(t)$.  We first collect some basic identities that will
aid in evaluating the limit:
\begin{lemma}
  \begin{align} 
    \int_{\Omega} \axf{2}(\br,\bro\o)    d \br &= \khat^2 \gamma   U(\bro,\brb) + 2 \psibar \sqrt{\abs{\Omega}} \notag \\
&\phantom{=} - \frac{\khat^2}{\abs{\Omega}} \int_{\Omega} U(\bro,\br') U(\br',\brb)   d \br'
     , \label{eq:w2int} \\
     \int_{\Omega} U(\bro,\br') U(\br',\brb)   d \br' &= 
     \lim_{s \to 0} \frac{U(\bro,\brb) - \laxphi{0}(\brb,s)}{s}, \label{eq:uConvInt} \\
     \int_{\Omega} \paren{U(\br,\brb)}^2   d \br &= 
     \lim_{s \to 0} \int_{\Omega} \lG(\brb,\br',s) U(\br',\brb)   d \br'. \label{eq:uSqInt}
  \end{align}
\end{lemma}
\begin{proof}
  The first identity follows immediately from the definition of
  $\axf{2}$~\eqref{eq:f2} and~\eqref{eq:28b}. In the right hand side
  of~\eqref{eq:uConvInt} we replace the $U$ terms with time integrals
  of $G$ by~\eqref{eq:29}, switch the order of integration, and
  evaluate the spatial integral using the semigroup property of $G$ to
  find that
  \begin{multline} \label{eq:uConvInt2}
    \int_{\Omega} U(\bro,\br') U(\br',\brb)   d \br' = \\
    \int_0^{\infty} \int_{t}^{\infty} \brac{G(\bro,\brb,s) - \frac{1}{\abs{\Omega}}}   ds   dt. 
  \end{multline}
  As
  \begin{multline*}
    \int_{t}^{\infty} \brac{G(\bro,\brb,s) - \frac{1}{\abs{\Omega}}}   ds = \\
    U(\bro,\brb) - \int_{0}^{t} \brac{G(\bro,\brb,s) - \frac{1}{\abs{\Omega}}}   ds,
  \end{multline*}
  recalling the definition of $\axphi{0}(\brb,s)$~\eqref{eq:53} we
  see that
  \begin{align*}
    \int_{\Omega} &U(\bro,\br') U(\br',\brb)   d \br' \\
    &=  \int_0^{\infty} \paren{ U(\bro,\brb) - \int_{0}^{t} \axphi{0}(\brb,s')   ds'}  dt \\
    &= \lim_{s \to 0} \int_0^{\infty} \paren{ U(\bro,\brb) - \int_{0}^{t} \axphi{0}(\brb,s')   ds'}  e^{-s t} dt.
  \end{align*}
  \eqref{eq:uConvInt} then follows by definition of the Laplace transform. 

  Finally, by~\eqref{eq:28b} we have that 
  \begin{equation*}
    \int_{\Omega} \lG(\brb,\br',s) U(\br',\brb)   d \br' = 
    \int_{\Omega} \lG_0(\brb,\br',s) U(\br',\brb)   d \br'.
  \end{equation*}
  Using~\eqref{eq:29}, we have that $\lim_{s \to 0} \lG_0(\brb,\br',s)
  = U(\brb,\br')$. A dominated convergence argument then
  implies~\eqref{eq:uSqInt}.
\end{proof}

We are now ready to evaluate the limit of $\axrho{2}(\br,t)$ as $t \to
\infty$. By dominated convergence and~\eqref{eq:29}, it is immediate
from~\eqref{eq:st2} that
\begin{align}
  \lim_{t \to \infty} \axrho{2}(\br,t) &= \frac{\khat^2 \gamma}{\abs{\Omega}} U(\brb,\bro)
  - \frac{\khat}{\abs{\Omega}} \int_0^{\infty} \axphi{1}(\brb,s)   ds \notag \\
  &\phantom{=}\quad - \frac{1}{\abs{\Omega}} \int_{\Omega} \axf{2}(\br',\bro)   d \br', \notag \\
  &=   - \frac{\khat}{\abs{\Omega}} \int_0^{\infty} \axphi{1}(\brb,s)   ds - \frac{2 \psibar}{\sqrt{\abs{\Omega}}} \notag \\
  &\phantom{=}\quad + \frac{\khat^2}{\abs{\Omega}^2} \int_{\Omega} U(\bro,\br') U(\br',\brb)   d \br', \label{eq:st2lim}
\end{align}
where the last line follows by~\eqref{eq:w2int}. By definition of the Laplace
transform, 
\begin{equation*}
  \int_0^{\infty} \axphi{1}(\brb,s)   ds= \lim_{s \to 0} \lim_{\br \to \brb} \laxphi{1}(\br,s).
\end{equation*}
From the definition of $\axphi{1}(\br,t)$~\eqref{eq:54} we find that 
\begin{align*}
  \laxphi{1}(\br,s) &= \khat \laxphi{0}(\brb,s) \brac{U(\br,\brb) - \lG_0(\br,\brb,s)} \\
  &\phantom{=}\quad +\frac{\khat}{\abs{\Omega} s} \brac{ U(\bro,\brb) - \laxphi{0}(\brb,s)} \\
    &\phantom{=}\quad  + \frac{\khat}{\abs{\Omega}} \int_{\Omega} \lG(\br,\br',s) U(\br',\brb)   d \br'.
\end{align*}
Substituting into~\eqref{eq:st2lim}, and using~\eqref{eq:uConvInt},
\eqref{eq:uSqInt}, and the definition of
$\psibar$~\eqref{eq:princEfuncNormalize} we find
\begin{multline} \label{eq:st2LimRep}
  \lim_{t \to \infty} \axrho{2}(\br,t) = -\frac{\khat^2 U(\brb,\bro)}{\abs{\Omega}} \\
   \times \lim_{s \to 0} \lim_{\br \to \brb} \brac{U(\br,\brb) - \lG_0(\br,\brb,s)}.
\end{multline}
The limit of the bracketed term can be evaluated by splitting $U$ and $\lG_0$ into
regular and singular parts (at $\br = \brb$). We write that
\begin{equation*}
  G_0(\br,\brb,t) = R_0(\br,\brb,t) + \frac{1}{\paren{4 \pi D t}^{3/2}} e^{-\abs{\br - \brb}/4 D t},
\end{equation*}
where $R_0(\brb,\brb,t)$ is finite as $t \to 0$. Using~\eqref{eq:29} we see that
\begin{equation*}
  U(\br,\brb) = \lR_0(\br,\brb,0) + \frac{1}{\khat \abs{\br - \brb}},
\end{equation*}
where $\lR_0(\br,\brb,s)$ denotes the Laplace transform of $R$. As such,
\begin{align*}
  \lim_{\br \to \brb} \brac{U(\br,\brb) - \lG_0(\br,\brb,s)} &= 
  \lR_0(\brb,\brb,0) \\
    &\phantom{=} - \lR_0(\brb,\brb,s) + \frac{1}{\khat} \sqrt{\frac{s}{D}}.
\end{align*}
Evaluating the $s$ limit in~\eqref{eq:st2LimRep}, it follows that as
$t \to 0$, $\axrho{2}(\br,t) \to 0$.

\section{Spherically-symmetric Neumann Green's function}
\label{S:neuGFAppendix} Let $g(r,r_0,t)$ denote the spherically
symmetric solution to the diffusion equation, satisfying
\begin{align*}
  \PD{g}{t} &= D \frac{1}{r^2} \PD{}{r} \brac{r^2 \PD{g}{r}}, \quad r \in \left[0,1 \right),\\
  \PD{g}{r} &= 0, \quad r = 1,
\end{align*}
with the initial condition that $g(r,r_0,0) = \delta(r-r_0) / r^2$. With this choice,
\begin{equation*}
  g(r,r_0,t) = \iint_{\partial B_{1}(\vec{0})} G(\br,\br_0,t)   dS,
\end{equation*}
for $\partial B_1(\vec{0})$ the boundary of the unit sphere.  Here
$G(\br,\br_0,t)$ denotes the solution to the corresponding three
dimensional diffusion equation~\eqref{eq:diffGreensFunct}.  Note also
the normalization that
\begin{equation*}
  \int_0^1 g(r,r_0,t) r^2   dr = 1 = \iiint_{\Omega} G(\br,\br_0,t)   d \br.
\end{equation*}
By eigenfunction expansion we find 
\begin{multline*} 
  g(r,r_0,t) = 3\\ + 2 \sum_{n=1}^{\infty} \paren{1 + \frac{\mu_n}{D}} \sinc (\sqrt{\frac{\mu_n}{D}} r)
  \sinc (\sqrt{\frac{\mu_n}{D}} r_0) e^{-\mu_n t},
\end{multline*}
where the eigenvalues $\mu_n$ satisfy~\eqref{eq:60} and we use the
convention that
\begin{equation*}
  \sinc(x) = \frac{\sin(x)}{x}.
\end{equation*}

\section{Numerics} \label{S:numericsPoints} When evaluating the series
for the exact solution~\eqref{eq:52} and asymptotic
approximation~\eqref{eq:13}, we sum until the magnitude of the last
added term drops below a given error threshold.  We used an error
threshold of $10^{-14}$ for the exact solution and $10^{-7}$ for the
uniform approximation.  The figures are generated with 1000
equally-spaced points for $10^{-3}\leq t \leq 1$ and 500 equally-spaced points for $1 < t < 30$.

\bibliographystyle{apsrev4-1.bst}

\end{document}